\documentclass[12pt]{article}
\usepackage[T1]{fontenc}
\usepackage{lmodern}

\usepackage{sectsty}
\allsectionsfont{\textsf}
\sectionfont{\centering\sffamily}
\subsectionfont{\centering\sffamily}
\subsubsectionfont{\centering\sffamily}

\usepackage{etoolbox}
\usepackage{mathptmx}
\patchcmd{\abstract}{\sfshape\abstractname}{{\sf \abstractname}}{}{}

\usepackage{tikz}
\usetikzlibrary{positioning, arrows.meta}

\usepackage{amssymb}
\usepackage{bm}
\usepackage{bbm}
\usepackage[cmex10]{amsmath}
\usepackage{empheq}
\usepackage{amsthm}
\usepackage{mathtools}
\usepackage[longnamesfirst]{natbib}
\usepackage{xcolor}
\usepackage[%
 setpagesize=false,%
 bookmarks=true,%
 bookmarksdepth=tocdepth,%
 bookmarksnumbered=true,%
 colorlinks=true,
 citecolor=black,
 urlcolor=blue,
 linkcolor=blue,%
 pdftitle={},%
 pdfsubject={},%
 pdfauthor={},%
 pdfkeywords={}%
]{hyperref}
\usepackage{hypernat}
\usepackage[margin= 1in]{geometry}
\usepackage{latexsym}
\usepackage{graphicx}
\usepackage[utf8]{inputenc}
\usepackage{booktabs}
\usepackage{multirow}

\theoremstyle{plain}

\theoremstyle{plain}
\newtheorem{proposition}{Proposition}[section]
\theoremstyle{plain}
\newtheorem{lemma}{Lemma}[section]
\theoremstyle{plain}

\theoremstyle{plain}
\newtheorem{fact}{Fact}[section]
\theoremstyle{plain}
\newtheorem{corollary}{Corollary}[section]
\theoremstyle{definition}
\newtheorem{remark}{Remark}[section]
\theoremstyle{definition}
\newtheorem{definition}{Definition}[section]
\theoremstyle{definition}

\theoremstyle{definition}

\theoremstyle{plain}

\begin{document}
\title{{\sf \LARGE Social Preferences and Cooperation: Beliefs, Robustness, and the Limits of Altruism}\thanks{First Draft: March 2021.}
}
\author{\textsc{Yosuke Hashidate\thanks{Affiliation: Faculty of Economics, Sophia University; Address: 7-1, Kioi-cho, Chiyoda-ku, Tokyo 102-8554, Japan; Email: {\tt hashidate@sophia.ac.jp}}}
}
\date{Current Draft: September 5, 2026} 
\maketitle
\begin{abstract}
We study a mechanism of cooperation in the Prisoner's Dilemma (PD). Incorporating social preferences as efficiency concerns into the PD game, we study how altruism translates into cooperation. Under complete information, cooperation requires the opponent's altruism to clear a threshold. We then introduce a subjective extension of Bayesian Nash equilibrium that relaxes the Common Prior Assumption, letting players hold heterogeneous, potentially misspecified beliefs about each other's altruistic type. Cooperation then depends on beliefs about altruism rather than altruism itself, and can be sustained even when opponents are, on average, only weakly altruistic. When fear of exploitation dominates the temptation to defect, beliefs about the opponent's cooperation become strategic complements, so a cooperative and an uncooperative equilibrium can coexist under identical payoffs and an identical, correctly specified prior. 
Using multiplier preferences, we then study how robust this belief-driven cooperation is to model misspecification. Cooperation is fragile: it survives only above a threshold level of confidence in one's own belief, and can unravel even when the belief itself correctly supports cooperation. As a formal extension, the same robust-control apparatus, applied to a player's action choice, nests Nash equilibrium, Bayesian Nash equilibrium, and logit Quantal Response Equilibrium as limiting cases. Cooperation depends less on how altruistic agents are than on what they believe about each other, and how confident they are that this belief is right.
\end{abstract}

\section{Introduction}
\label{sec:introduction}

The Prisoner's Dilemma is a canonical model of strategic interaction in which individually rational behavior leads to a socially inefficient outcome: mutual defection is the unique Nash equilibrium even though mutual cooperation is Pareto efficient. Yet cooperation in one-shot and finitely repeated Prisoner's Dilemma games is pervasive, both in the laboratory and in the field. A large literature explains part of this gap by relaxing the assumption of narrow self-interest: models of efficiency concerns and inequity aversion \citep{FS_1999,CR_2002} show that a player who partly internalizes her partner's payoff may rationally choose to cooperate. Empirically, the resulting cooperation rate is systematically related to the payoff structure $(R,P,S,T)$ of the game: \citet{R_1967} and \citet{AOSSW_2001} propose cooperation indices built from the greed/temptation term $T-R$ and the fear/risk term $P-S$; \citet{M_2018} and \citet{GLSW_2024} confirm and refine these patterns meta-analytically; and \citet{HKMT_2021} show that the underlying distributional preferences are stable within a person across games and time.

\begin{table}[h]
\label{tab:PD}
\centering
\caption{A Prisoner's Dilemma Game of $(R,P,S,T)$ with $T>R>P>S$}
\begin{tabular}{l|cc}
Player 1 / Player 2 & Cooperation ($C$) & Defection ($D$) \\\hline
Cooperation ($C$) & $R,R$ & $S,T$ \\
Defection ($D$) & $T,S$ & $P,P$
\end{tabular}
\end{table}

This literature, however, almost universally treats altruism as \emph{common knowledge}: when a model allows player $i$ to be uncertain about player $j$'s type, it is standard to further assume that both players share the same prior over that type --- the common prior assumption (CPA). The assumption is analytically convenient, but it is also behaviorally arbitrary. Two players who have never played together, who have met each other only through a handful of noisy interactions, or who simply hold different life experiences of how generous strangers tend to be, have no obvious reason to converge on an identical assessment of a third party's --- or of each other's --- altruism. \citet{M_1995} makes exactly this objection to the common prior assumption in economic theory more broadly, and Aumann's own foundational work on correlated equilibrium \citep{A_1974,A_1987} already distinguishes an \emph{objective} equilibrium concept, which requires a common prior, from a \emph{subjective} one, which does not. If beliefs about altruism are genuinely subjective, cooperation cannot be fully understood by asking how altruistic people are; one must also ask what people \emph{believe} about how altruistic others are, and how confident they are in that belief.

This paper develops a framework for asking exactly that question in the Prisoner's Dilemma. We embed a standard efficiency-concerns model of social preferences into the one-shot PD and then relax common knowledge of types along two complementary dimensions. 
First, we build a genuine incomplete-information game in which each player privately knows her own altruism and holds a subjective --- not necessarily common, not necessarily correct --- belief about her opponent's, in the spirit of Aumann's subjective correlated equilibrium and of \citet{KL_1993}'s subjective equilibrium concept for learning in games. 
Second, while subjective beliefs mediate cooperation, such beliefs are inherently vulnerable to misspecification and doubt in unfamiliar environments. We therefore introduce the multiplier-preferences framework of robust control \citep{HS_2001, S_2011} to capture an agent's epistemic confidence in her own probabilistic model. 
The resulting claim, which organizes the rest of the paper, is that cooperation in the Prisoner's Dilemma is sustained not by altruism itself, but by what players believe about each other's altruism, and by how robust that belief is to being mistaken.


{\small 
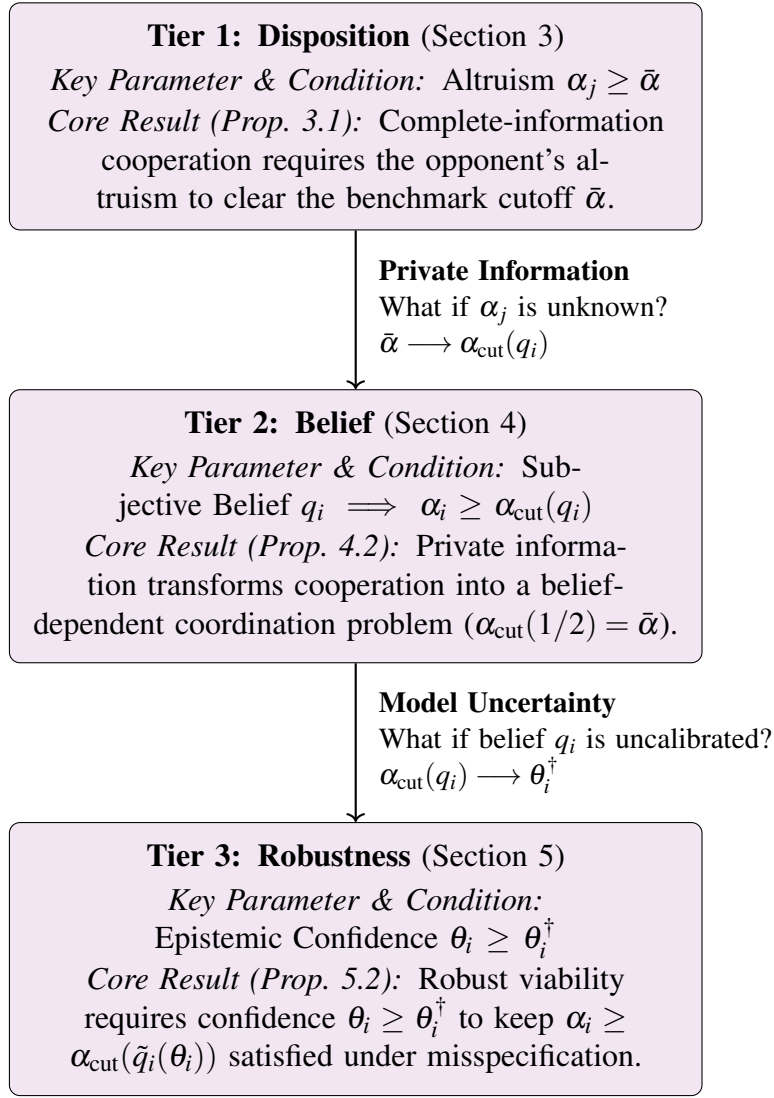
\begin{figure}[htbp]
\centering
\begin{tikzpicture}[
  box/.style={rectangle, rounded corners, draw=black, fill=violet!10,
    text width=8.6cm, align=center, minimum height=2.3cm, inner sep=8pt},
  arrlabel/.style={align=left, font=\small, text width=5.5cm}
]
\node[box] (tier1) {
  \textbf{Tier 1: Disposition} (Section 3) \\[2pt]
  \textit{Key Parameter \& Condition:} Altruism $\alpha_j \ge \bar{\alpha}$ \\[1pt]
  \textit{Core Result (Prop.\ 3.1):} Complete-information cooperation requires the opponent's altruism to clear the benchmark cutoff $\bar{\alpha}$.
};
 
\node[box, below=2.1cm of tier1] (tier2) {
  \textbf{Tier 2: Belief} (Section 4) \\[2pt]
  \textit{Key Parameter \& Condition:} Subjective Belief $q_i \implies \alpha_i \ge \alpha_{\mathrm{cut}}(q_i)$ \\[1pt]
  \textit{Core Result (Prop.\ 4.2):} Private information transforms cooperation into a belief-dependent coordination problem ($\alpha_{\mathrm{cut}}(1/2) = \bar{\alpha}$).
};
 
\node[box, below=2.1cm of tier2] (tier3) {
  \textbf{Tier 3: Robustness} (Section 5) \\[2pt]
  \textit{Key Parameter \& Condition:} Epistemic Confidence $\theta_i \ge \theta_i^\dagger$ \\[1pt]
  \textit{Core Result (Prop.\ 5.2):} Robust viability requires confidence $\theta_i \ge \theta_i^\dagger$ to keep $\alpha_i \ge \alpha_{\mathrm{cut}}(\tilde{q}_i(\theta_i))$ satisfied under misspecification.
};
 
\draw[->, thick] (tier1) -- (tier2)
  node[arrlabel, midway, right=0.15cm] {\textbf{Private Information} \\ What if $\alpha_j$ is unknown? \\ $\bar{\alpha} \longrightarrow \alpha_{\mathrm{cut}}(q_i)$};
 
\draw[->, thick] (tier2) -- (tier3)
  node[arrlabel, midway, right=0.15cm] {\textbf{Model Uncertainty} \\ What if belief $q_i$ is uncalibrated? \\ $\alpha_{\mathrm{cut}}(q_i) \longrightarrow \theta_i^\dagger$};
 
\end{tikzpicture}
\caption{The Three-Tier Epistemic Hierarchy of Cooperation. The diagram illustrates the progressive generalization of cooperation thresholds from complete-information disposition ($\bar{\alpha}$), to belief-dependent cutoffs ($\alpha_{\mathrm{cut}}(q_i)$), and finally to the epistemic confidence threshold ($\theta_i^\dagger$) that guards against model misspecification. A failure at any tier causes cooperation to collapse.}
\label{fig:hierarchy}
\end{figure}
}

Three results substantiate this claim. 
First, working with a fully specified Bayesian game with private altruism types and equilibrium cutoff strategies, we show that cooperation can be sustained purely through mutual beliefs even when actual altruism, on average, is weak --- belief, not disposition, does the work (Section~\ref{sec:bayesian}). 
Second, and more surprisingly, private information about altruism can transform cooperation into a coordination problem: we show that when the payoff structure makes fear/risk salient relative to greed/temptation, beliefs about whether the opponent will cooperate become \emph{strategic complements}, so that for the identical payoff structure and an identical, correctly specified common prior, there can exist both a fully cooperative and a fully uncooperative equilibrium, each self-fulfilling (Proposition~\ref{prop:symmetric-equilibrium}). Cooperation, in this regime, is a coordination problem as much as an altruism problem. To our knowledge, this is the first demonstration that private information about a purely distributional social preference --- rather than reputation or repeated-game punishment --- can by itself generate self-fulfilling multiplicity in the one-shot Prisoner's Dilemma; we return to this point in Section~\ref{sec:bayesian}, which is this paper's main theoretical contribution.\footnote{We provide a sharp visual characterization of this phase transition in Figure~\ref{fig:cutoff-br}, illustrating how the same environment tips from a unique outcome into belief-driven multiplicity depending on the payoff structure.} 
Third, while this multiplicity underscores the primacy of beliefs over dispositions, it immediately raises a deeper epistemic challenge: \emph{how reliable are these beliefs themselves?} In one-shot or unfamiliar interactions, an agent's assessment of an opponent's altruism is rarely grounded in objective data; rather, it is inherently subjective, fragile, and subject to model misspecification. To formalize how confidence in one's own model shapes behavior, we introduce a robust-control formulation (interpreted as \emph{epistemic confidence}) and show that a player's concern for misspecification can undo cooperation that her baseline reference belief would otherwise support. Specifically, even if the reference distribution of the opponent's altruism predicts cooperation, sufficiently strong doubt pessimistically distorts the perceived likelihood of cooperation: below a sharp viability threshold on the epistemic confidence parameter, $\theta_i^\dagger$, belief-driven cooperation unravels (Proposition~\ref{prop:robustness-threshold}). We further show that this robust-control apparatus, applied to a player's action choice rather than her belief, nests the logit Quantal Response Equilibrium of \citet{MP_1995} --- whose entropy-regularized microfoundation parallels the rational-inattention foundation for logit choice in \citet{MM_2015} --- and, in the appropriate limits, the ordinary Nash and Bayesian Nash equilibria of the earlier sections, unifying all of the paper's solution concepts within a single framework (Section~\ref{sec:robust}).
 
These results speak to a growing body of work suggesting that beliefs, not just preferences, drive cooperative behavior. Existing social-preference models are, almost by construction, models of preferences: they ask how altruistic, inequity-averse, or reciprocal a person is, and predict cooperation accordingly. However, if we strictly adhere to the Common Prior Assumption, all uncertainty about others must reduce to a shared, objective distribution, leaving no room for subjective differences in how people view the strategic environment. Recent evidence suggests that beliefs about cooperation carry independent explanatory power beyond such preference measures. Using data from indefinitely repeated Prisoner's Dilemma play, \citet{GR_2024} show that beliefs about a partner's likely cooperation are central to predicting behavior, strongly predict cooperation, and evolve systematically --- becoming more accurate --- with experience. Our theoretical results offer one account of \emph{why} this might be so: if cooperation is genuinely belief-dependent, and if beliefs are heterogeneous and subject to model misspecification, then variation in cooperation rates across otherwise identical strategic environments need not reflect variation in underlying altruism at all. We also connect this to the finding of \citet{CZ_2025} that people can behave more morally under uncertainty, and to the moral wiggle room and excuse-driven-selfishness literatures \citep{DWK_2007, E_2016}, which instead emphasize how uncertainty can be exploited to justify selfish behavior; Section~\ref{subsec:uncertainty-cooperation} discusses how our robust-control mechanism differs from both.
Existing studies emphasize beliefs about future actions. We instead study beliefs about the underlying social preference types that generate those actions.

This paper also contributes to the growing literature that incorporates non-expected utility and model misspecification into strategic interactions. Recent axiomatic work by \cite{MPST_2024} completely characterizes monotone additive statistics, providing a robust mathematical foundation for evaluating non-expected utility over lotteries. Building on this, \cite{SSTW_2025} axiomatize the Statistic Response Equilibrium (SRE), demonstrating that if players evaluate independent games separately (the ``bracketing'' axiom) and respond monotonically, their equilibrium behavior must be governed by such statistics. 

Our work bridges these abstract axiomatic foundations with applied strategic settings. By operationalizing their framework---specifically, by applying multiplier preferences (a concrete instance of monotone additive statistics) combined with logit choice noise to a social dilemma---we uncover the precise mechanism through which cooperation collapses. Furthermore, our approach complements the literature on learning under misspecification, such as the Berk-Nash equilibrium \citep{EP_2016, BH_2025}. While the Berk-Nash framework elegantly captures the long-run steady state of learning with misspecified models, our Robust Logit QRE captures the interim, self-defensive dynamics where players pessimistically tilt their beliefs out of fear of being exploited.
Our contribution is not to provide a new equilibrium concept. Rather, we show how these recent non-expected utility foundations illuminate a concrete question in social dilemmas: when social preferences do and do not translate into cooperation.

The remainder of the paper is organized as follows. Section~\ref{sec:distributional-preferences} specifies the efficiency-concerns model of social preferences used throughout. Section~\ref{sec:nash} derives the complete-information benchmark: Nash equilibrium behavior when altruism is common knowledge, and its comparative statics in the payoff structure and in the opponent's altruism. Section~\ref{sec:bayesian} is the paper's main theoretical contribution: it builds the incomplete-information Bayesian game, derives cutoff equilibria in closed form, and establishes the multiplicity result described above. Section~\ref{sec:robust} adds the third and final tier of this hierarchy --- disposition, belief, and now robustness --- introducing robustness to model misspecification via multiplier preferences and deriving the threshold that determines whether belief-driven cooperation survives; as a formal extension, it further shows that the same apparatus, applied to a player's action choice, nests a Robust Logit Quantal Response Equilibrium encompassing Nash, Bayesian Nash, and standard logit QRE as limiting cases. Section~\ref{sec:discussion} discusses the interaction between payoff structures and belief structures and returns to the role of uncertainty in sustaining or undermining cooperation. Section~\ref{sec:conclusion} concludes.

\section{Distributional Preferences}
\label{sec:distributional-preferences}
Before analyzing equilibrium behavior in the Prisoner's Dilemma $\mathcal{G}_{PD}$ (Table~\ref{tab:PD}), we specify the class of distributional preferences that will be used throughout the paper. Let $I=\{1,2\}$ be the set of players, and let $\mathbf{x}=(x_1,x_2)\in\mathbb{R}^2$ denote an allocation of monetary payoffs across players, i.e., the raw stage-game payoffs listed in Table~1. We model social preferences as \emph{efficiency concerns}: each player partially internalizes the payoff of the opponent.

\begin{definition}[Efficiency Concerns]
\label{def:efficiency-concerns}
A preference relation $\succsim_i$ on $\mathbb{R}^2$ has a model of \emph{efficiency concerns} if there exists a nonnegative real number $\alpha_i$ such that $\succsim_i$ is represented by the function $u_i:\mathbb{R}^2\to\mathbb{R}$ defined by
\begin{align*}
    u_i(\mathbf{x}) = u_i(x_i,x_j) = x_i + \alpha_i x_j,
\end{align*}
where $i,j\in I,\ i\neq j$.
\end{definition}
The parameter $\alpha_i$ captures player $i$'s \emph{pure altruism}: the larger $\alpha_i$ is, the more player $i$ is willing to trade off own payoff $x_i$ against the opponent's payoff $x_j$. This is a simplified, single-parameter version of the efficiency-concerns models used in the distributional-preferences literature \citep[e.g.,][]{CR_2002}, chosen for tractability since our focus is on the interaction between altruism and \emph{beliefs} about the opponent's altruism, rather than on the functional form of social preferences per se.

We restrict attention throughout to $\alpha_i\ge 0$; this rules out spiteful preferences and lets us interpret $\alpha_i$ as a measure of how altruistic, as opposed to purely self-interested, player $i$ is. An alternative, inequity-averse specification of distributional preferences \citep{FS_1999} is developed as a robustness check in Appendix~\ref{app:inequity-aversion}.

Having fixed the preference structure, we now ask what it implies for behavior. We proceed in three steps of increasing realism. Section~\ref{sec:nash} treats altruism as common knowledge and derives the resulting Nash equilibria; this is the standard case in the social-preferences literature and serves as our benchmark. Sections~\ref{sec:bayesian} and~\ref{sec:robust} then argue that this benchmark rests on an assumption --- that each player knows the other's altruism exactly --- that is rarely defensible outside the laboratory, and ask what happens to cooperation once it is relaxed.

\section{Nash Equilibrium and Cooperation}
\label{sec:nash}

\subsection{Nash Equilibrium}
\label{subsec:nash-equilibrium}
We set up the underlying game form. Let $A_i$ denote the action set of player $i\in I$, with typical element $a_i\in A_i$, and let $g_i:A\to\mathbb{R}$ be player $i$'s payoff function, where $A=A_1\times A_2$. We consider finite games, represented in reduced form as $\mathcal{G}=\langle I,(A_i)_{i\in I},(g_i)_{i\in I}\rangle$.

We study the one-shot Prisoner's Dilemma $\mathcal{G}_{PD}$ described in Table~\ref{tab:PD}, with $A_1=A_2=\{C,D\}$ and payoff parameters $(R,P,S,T)$ satisfying $T>R>P>S$. We further assume $2R>T+S$, which guarantees that the cooperative outcome $(C,C)$ is Pareto efficient relative to the alternating exploitation outcomes.

Embedding the efficiency-concerns preferences of Definition~\ref{def:efficiency-concerns} into $\mathcal{G}_{PD}$, player $i$'s payoff becomes $u_i(\bm{x})=x_i+\alpha_i x_j$ for $i\neq j$: each player's own altruism parameter weights her own valuation of the partner's raw payoff.\footnote{An earlier draft of Table~\ref{tab:PD_dist} had the two coordinates' altruism subscripts transposed (player $i$'s payoff weighted by $\alpha_j$, player $j$'s by $\alpha_i$). This is corrected below; own-type-weights-own-utility is both the reading suggested by Definition~\ref{def:efficiency-concerns}'s description and the convention that all of the downstream results (Proposition~\ref{prop:nash-efficiency}, Corollaries~\ref{cor:payoff-comparative-statics}--\ref{cor:altruism-comparative-statics}, and Proposition~\ref{prop:cooperation-model-uncertainty} in Section~\ref{sec:robust}) already verify against; the equilibrium formulas themselves are unaffected; only Table~2's cell labels change. This is also the convention used directly in Section~\ref{sec:bayesian}, where each player's own type is private information parametrizing her own utility.} This yields the modified game shown in Table~\ref{tab:PD_dist}.

\begin{table}[h]
\label{tab:PD_dist}
\centering
\caption{The Prisoner's Dilemma Game $\mathcal{G}_{PD}$ with Efficiency Concerns}
\begin{tabular}{l|cc}
Player $i$ / Player $j$ & Cooperation ($C$) & Defection ($D$) \\\hline
Cooperation ($C$) & $(1+\alpha_i)R,\ (1+\alpha_j)R$ & $S+\alpha_i T,\ T+\alpha_j S$ \\
Defection ($D$) & $T+\alpha_i S,\ S+\alpha_j T$ & $(1+\alpha_i)P,\ (1+\alpha_j)P$
\end{tabular}
\end{table}

Let $\Delta(A_i)$ denote the set of mixed strategies for player $i$, with typical element $\rho_i\in\Delta(A_i)$, and let $\rho=(\rho_i)_{i\in I}$. Player $i$'s expected payoff is
\[
\mathbb{E}_\rho[g_i] := \sum_{a_i\in A_i}\prod_{i\in I}\rho_i(a_i)\,g_i(a_i,a_j).
\]
Player $i$'s best-response correspondence given $\rho_{-i}$ is
\[
BR_i(\rho_{-i}) := \{\rho_i\in\Delta(A_i)\ |\ \forall \rho_i'\in\Delta(A_i),\ \mathbb{E}_{\rho_i,\rho_j}[g_i]\ge \mathbb{E}_{\rho_i',\rho_j}[g_i]\}.
\]

\begin{definition}[Nash Equilibrium]
A mixed strategy profile $\rho^*\in\prod_{i\in I}\Delta(A_i)$ is a \emph{Nash equilibrium} if $\rho_i^*\in BR_i(\rho_{-i}^*)$ for all $i\in I$.
\end{definition}

\begin{fact}[Existence]
Every finite game $\mathcal{G}$ has at least one Nash equilibrium.
\end{fact}

The next proposition characterizes the Nash equilibria of $\mathcal{G}_{PD}$ under efficiency concerns.

\begin{proposition}[Nash Equilibria and Efficiency Concerns]
\label{prop:nash-efficiency}
Given $\mathcal{G}_{PD}$, suppose each $\succsim_i$ has a model of efficiency concerns (Definition~\ref{def:efficiency-concerns}). Then there is a Nash equilibrium $\rho^*=(\rho_i^*)_{i\in I}$. For each $i,j\in I$ with $i\neq j$:

\smallskip
\noindent \emph{If $T-R\ge P-S$:}
\begin{itemize}
\item If $\alpha_j \le \dfrac{P-S}{T-P}$, then $\rho_i^*(C,\{C,D\})=0$.
\item If $\alpha_j \in \left(\dfrac{P-S}{T-P},\ \dfrac{T-R}{R-S}\right)$, then
\[
\rho_i^*(C,\{C,D\}) = \frac{P-S-\alpha_j(T-P)}{(1+\alpha_j)\big(-(T-R)+P-S\big)}.
\]
\item If $\alpha_j \ge \dfrac{T-R}{R-S}$, then $\rho_i^*(C,\{C,D\})=1$.
\end{itemize}

\smallskip
\noindent \emph{If $T-R\le P-S$:}
\begin{itemize}
\item If $\alpha_j \le \dfrac{T-R}{R-S}$, then $\rho_i^*(C,\{C,D\})=1$.
\item If $\alpha_j \in \left(\dfrac{T-R}{R-S},\ \dfrac{P-S}{T-P}\right)$, then
\[
\rho_i^*(C,\{C,D\}) = \frac{P-S-\alpha_j(T-P)}{(1+\alpha_j)\big(-(T-R)+P-S\big)}.
\]
\item If $\alpha_j \ge \dfrac{P-S}{T-P}$, then $\rho_i^*(C,\{C,D\})=0$.
\end{itemize}
\end{proposition}

\begin{proof}
See Appendix~\ref{app:proof-prop-nash-efficiency}.
\end{proof}

The interpretation is as follows. Fix player $i$ and suppose $T-R\ge P-S$, i.e., the effect of greed/temptation exceeds that of fear/risk. If the opponent's altruism $\alpha_j$ is low (below $(P-S)/(T-P)$), player $i$ defects with certainty; if $\alpha_j$ is high (above $(T-R)/(R-S)$), player $i$ cooperates with certainty; in between, $\rho_i^*(C)$ rises continuously from $0$ to $1$, as plotted in Figure~\ref{fig:cooperation-rates}. When $T-R\le P-S$ instead, the same continuous transition occurs but in the opposite direction: low $\alpha_j$ sustains cooperation, and high $\alpha_j$ triggers defection. In each case, if \emph{both} players' altruism levels place them in the pure-cooperation region, the profile $(a_1^*,a_2^*)=(C,C)$ is an equilibrium; if both are in the pure-defection region, $(D,D)$ is an equilibrium (which includes the fully selfish case $\alpha_i=0$ for all $i$).

\begin{figure}[h]
\centering
\includegraphics[width=0.65\textwidth]{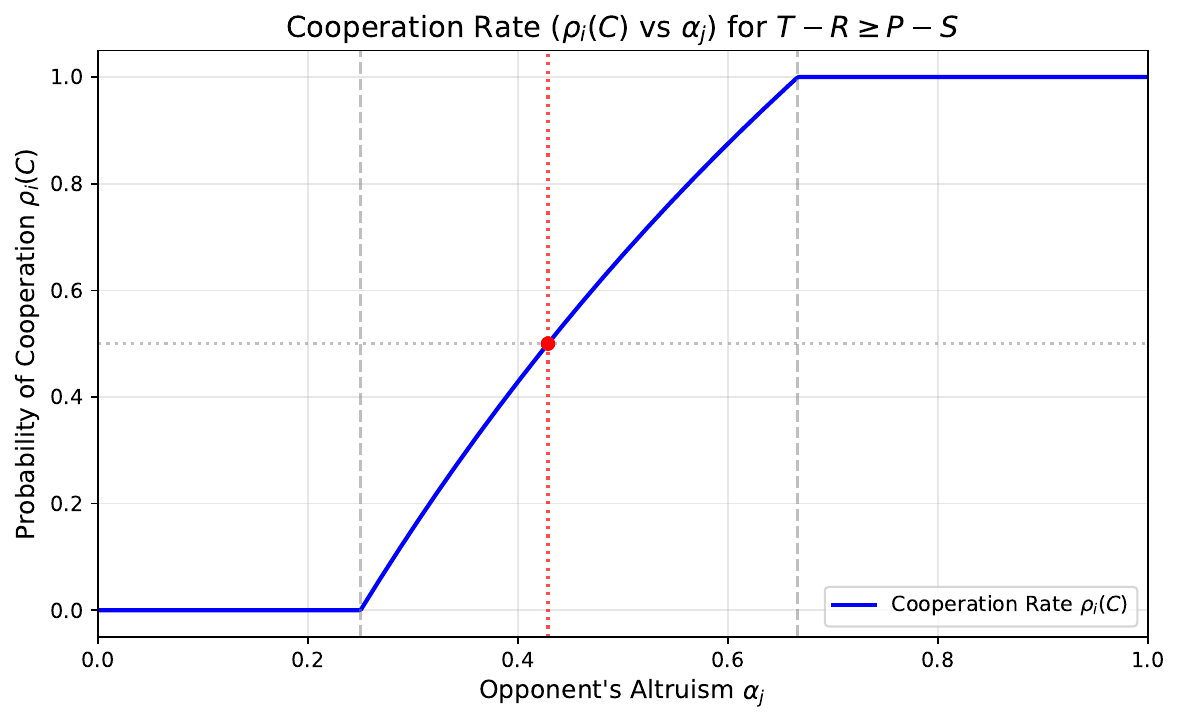}
\caption{Cooperation Rates with the Case of $T-R\ge P-S$. Let $\overline\alpha=\frac{T-R+P-S}{T+R-P-S}$. By Proposition~\ref{prop:nash-efficiency} and Corollary~\ref{cor:altruism-comparative-statics}, $\rho_i(C)$ rises monotonically from $0$ at $\alpha_j=(P-S)/(T-P)$ to $1$ at $\alpha_j=(T-R)/(R-S)$, crossing $1/2$ at $\overline\alpha$. Section~\ref{subsec:robust-bne} shows that $\mu_i(\{\alpha_j\ge\overline\alpha\})\ge\mu_i(\{\alpha_j<\overline\alpha\})$ implies $\rho_i(C)\ge\rho_i(D)$. This conceptual graph was created using Gemini 3.1 Pro.The author reviewed and edited the content as needed and takes full responsibility for the content of the graph.}
\label{fig:cooperation-rates}
\end{figure}

\subsection{Effects of Payoff Structures}
\label{subsec:payoff-effects}
Next, we ask how the payoff structure $(R,P,S,T)$ shapes the (interior, mixed-strategy) cooperation probability $\rho_i(C,\{C,D\})$ characterized in Proposition~\ref{prop:nash-efficiency}.

\begin{corollary}[Comparative Statics with Efficiency Concerns]
\label{cor:payoff-comparative-statics}
Given $\mathcal{G}_{PD}$, suppose each $i\in I$ has efficiency-concern preferences. Then, on the interior (mixed-strategy) region:
\begin{enumerate}
\item[(1)] (Effects of $R$) $\ \rho_i(C,\{C,D\})$ increases with $R$ if $\alpha_j\ge \dfrac{P-S}{T-P}$, and decreases otherwise.
\item[(2)] (Effects of $P$) $\ \rho_i(C,\{C,D\})$ increases with $P$ if $\alpha_j\ge \dfrac{T-R}{R-S}$, and decreases otherwise.
\item[(3)] (Effects of $T$) $\ \rho_i(C,\{C,D\})$ increases with $T$ if $\alpha_j\le \dfrac{P-S}{R-S}$, and decreases otherwise.
\item[(4)] (Effects of $S$) $\ \rho_i(C,\{C,D\})$ increases with $S$ if $\alpha_j\le \dfrac{T-R}{T-P}$, and decreases otherwise.
\end{enumerate}
\end{corollary}

\begin{proof}
See Appendix~\ref{app:proof-cor-payoff}.
\end{proof}
Within the mixed-strategy region identified in Proposition~\ref{prop:nash-efficiency}, the two case-defining thresholds in Corollary~\ref{cor:payoff-comparative-statics} are automatically satisfied for $R$ and $P$: whenever $T-R\ge P-S$, the relevant mixed region has $\alpha_j\ge (P-S)/(T-P)$ throughout, so $\rho_i$ unambiguously increases with $R$; whenever $T-R\le P-S$, the mixed region has $\alpha_j\ge (T-R)/(R-S)$ throughout, so $\rho_i$ unambiguously increases with $P$.\footnote{The converse case (decreasing in $R$, respectively $P$) is attained on the other side of the case split, i.e., on the mixed region associated with the opposite ordering of $T-R$ versus $P-S$.} Intuitively, more efficient mutual cooperation ($R\uparrow$) or a harsher mutual-defection outcome ($P\uparrow$, all else equal, raising the relative attractiveness of coordinating on $C$) makes cooperation easier to sustain in equilibrium.

The remaining two results are, at first glance, less intuitive. Increasing the temptation payoff $T$ tempts each player toward $D$; but Corollary~\ref{cor:payoff-comparative-statics}(3) shows that $\rho_i$ can still \emph{increase} with $T$ when the opponent's altruism $\alpha_j$ is low enough. The mechanism, spelled out in the proof, is that raising $T$ (or $S$) lowers the altruism threshold $\overline{\alpha}_j(R,P,T,S):=(P-S)/(T-P)$ above which player $j$ is willing to cooperate against a defecting $i$; since $j$ then becomes more willing to cooperate, $i$ is in turn more willing to cooperate as well, provided $\alpha_j$ is not already so large that the direct temptation effect dominates.

\paragraph{The Cooperation Index.} \citet[p.\ 151]{AOSSW_2001} propose an empirical cooperation index
\begin{align*}
    \mathcal{K}(R,P,T,S) := a + b\,\frac{T-R}{T-S} + c\,\frac{P-S}{T-S},
\end{align*}
where $a\in(0,1),\ b,c\in(-1,0)$, which is decreasing in both the greed/temptation term $(T-R)/(T-S)$ and the fear/risk term $(P-S)/(T-S)$. Corollary~\ref{cor:payoff-comparative-statics} recovers this pattern once $\alpha_i$ is sufficiently large for all $i\in I$: e.g., for $(R,P,T,S)=(1,0,3,-1)$, if $\alpha_i\ge 2/3$ for each $i$, our model is consistent with both \citet{AOSSW_2001} and the meta-analysis of \citet{M_2018}, which finds that fear/risk, $(P-S)/P$, has significant explanatory power for one-shot PD cooperation rates. \citet{CRR_2016} directly test the comparative static in $R$: holding $(P,T,S)=(2,7,1)$ fixed and varying $R\in\{3,4,5,6\}$, they find cooperation rates of $23\%,34\%,51\%,60\%$ respectively --- monotonically increasing in the social surplus $R$, consistent with Corollary~\ref{cor:payoff-comparative-statics}(1).

\subsection{Effects of Altruism}
\label{subsec:altruism-effects}

\begin{corollary}[Comparative Statics of Opponent's Altruism]
\label{cor:altruism-comparative-statics}
Given $\mathcal{G}_{PD}$, suppose each $\succsim_i$ has a model of efficiency concerns. Then,
\begin{align*}
    \frac{\partial}{\partial \alpha_j}\big[\rho_i^*(C,\{C,D\})\big] = \frac{T-S}{(1+\alpha_j)^2\big(T-R-(P-S)\big)} 
    \begin{cases} \ge 0 & \text{if } T-R\ge P-S, \\ 
        <0 & \text{if } T-R< P-S. 
    \end{cases}
\end{align*}
\end{corollary}

\begin{proof}
See Appendix~\ref{app:proof-cor-altruism}.
\end{proof}

Since $T-S>0$ always, the sign of this derivative is pinned down entirely by whether greed/temptation ($T-R$) or fear/risk ($P-S$) dominates the payoff structure. When $T-R\ge P-S$, a more altruistic opponent makes cooperation \emph{more} attractive --- the standard, intuitive case. When $T-R< P-S$, however, a more altruistic opponent makes cooperation \emph{less} attractive: because $\rho_i^*$ is pinned down by keeping the opponent indifferent, a highly altruistic $j$ is, in this payoff regime, easiest to sustain in cooperation precisely when $i$ cooperates \emph{less} often. The crossing point $\overline{\alpha}=\dfrac{T-R+P-S}{T+R-P-S}$, where $\rho_i^*=\tfrac12$, will reappear throughout the rest of the paper as the single statistic that separates altruism levels that sustain cooperation from those that do not.

This section presupposes that $\alpha_j$ is common knowledge --- that player $i$ knows exactly how altruistic her opponent is. In almost any real strategic encounter, a player has, at best, a noisy and personal impression of her partner's generosity, formed from limited experience, reputation, or cheap talk, and there is no reason two players should hold the \emph{same} impression of a third party or of each other. Sections~\ref{sec:bayesian} and~\ref{sec:robust} take this seriously: they replace common knowledge of $\alpha_j$ with a subjective belief about it, and ask whether --- and how --- cooperation survives the switch.

\section{Bayesian Nash Equilibrium and Cooperation}
\label{sec:bayesian}

\subsection{Bayesian Nash Equilibrium}
\label{subsec:bne}

Section~\ref{sec:nash} treated $\alpha_j$ as commonly known. We now treat each player's altruism as private information. Let $\mathcal{A}_i := \left\{\alpha_i>0 \ \middle|\ \alpha_i\in\left[\dfrac{P-S}{T-P},\dfrac{T-R}{R-S}\right]\right\}$ be player $i$'s type space --- the interval on which Proposition~\ref{prop:nash-efficiency}'s mixed-strategy equilibrium is supported --- and, for the symmetric case, let $\mathcal{A}:=\mathcal{A}_i=\mathcal{A}_j$. Let $\mu_i\in\Delta(\mathcal{A}_j)$ denote player $i$'s prior about the opponent's type.

Under the \emph{common prior assumption} (CPA), $\mu_i=\mu_j=\mu\in\Delta(\mathcal{A}^2)$ for all $i,j$. A strategy profile $f=(f_i)_{i\in I}$, with $f_i:\mathcal{A}_i\to\Delta(A_i)$ a type-contingent act, is a \emph{Bayesian Nash equilibrium} if
\begin{align*}
    \mathbb{E}_\mu\big[\widehat{g}_i(f^*,\alpha)\big] \ge \mathbb{E}_\mu\big[\widehat{g}_i(f_i,f_j^*,\alpha)\big]
\end{align*}
for all $i \in I$ and $f_i$, where $\widehat{g}_i:A\times\mathcal{A}^I\to\mathbb{R}$ is the type-dependent payoff function of Table~\ref{tab:PD_dist} and $\alpha=(\alpha_i)_{i\in I}$.

The CPA is a strong assumption: it requires that players agree on a single, correctly specified distribution over types. We relax it as follows.
\begin{definition}[Subjective Bayesian Nash Equilibrium]
\label{def:subjective-bne}
Given the Bayesian game $(\mathcal{G},\mathcal{A}^I)$, a strategy profile $f^*$ is a \emph{subjective Bayesian Nash equilibrium} if for all $i\in I$, $\alpha_i\in\mathcal{A}_i$, and $f_i'$, there exists $\mu_i\in\Delta(\mathcal{A}_j)$ such that
\begin{align*}
    \mathbb{E}_{\mu_i}\big[f^*(\alpha)\mid\alpha_j\big] \ge \mathbb{E}_{\mu_i}\big[f_i'(\alpha_j),f_j^*(\alpha_i)\mid\alpha_j\big].
\end{align*}
\end{definition}

Definition~\ref{def:subjective-bne} drops the requirement $\mu_i=\mu_j$ for all $i, j \in I$ with $i \neq j$, allowing players to hold heterogeneous, and potentially misspecified, beliefs about each other's altruism. This connects our equilibrium concept to subjective correlated equilibrium \citep{A_1974,A_1987}: players need not share a common view of the strategic environment, only best-respond given their own subjective probability assessment.

\subsection{Effects of Belief Structures}
\label{subsec:belief-effects}
Definition~\ref{def:subjective-bne} is a genuine incomplete-information game: each player privately knows her own type and best-responds to a belief about the opponent's type, with $u_i(\mathbf{x})=x_i+\alpha_i x_j$ as in Table~\ref{tab:PD_dist}. Because each player's own altruism now parametrizes her own utility, all the strategically relevant information a player needs about her opponent --- as Lemma~\ref{lem:cutoff} below shows --- reduces to a single aggregated belief about the opponent's action.

\subsubsection{A Cutoff Characterization}
Let $q_i:=\mathbb{E}_{\mu_i}[f_j(\alpha_j)]$ denote the (subjective) probability that player $i$ assigns to the opponent cooperating, aggregating over both $i$'s belief about $\alpha_j$ and $j$'s type-contingent strategy. Player $i$'s expected payoffs from $C$ and $D$, given own type $\alpha_i$, are
\begin{align*}
    U_i(C;\alpha_i) = q_i(1+\alpha_i)R + (1-q_i)(S+\alpha_i T),
\end{align*}
and
\begin{align*}
    U_i(D;\alpha_i) = q_i(T+\alpha_i S) + (1-q_i)(1+\alpha_i)P.
\end{align*}
Crucially, $U_i(C;\alpha_i)-U_i(D;\alpha_i)$ does not depend on $\alpha_j$ except through the scalar $q_i$: all the strategically relevant information about the opponent collapses to a single cooperation probability.

Since the coefficient of $\alpha_i$ in $U_i(C;\alpha_i)-U_i(D;\alpha_i)$ equals $(T-P)-q_iK$, where $K:=T-R-(P-S)$; this is strictly positive for all $q_i\in[0,1]$, it equals $T-P>0$ at $q_i=0$ and $R-S>0$ at $q_i=1$, and is linear in $q_i$). We obtain the following result.

The complete-information threshold $\bar{\alpha}$ from Section~\ref{sec:nash} reappears here as a natural benchmark. When an agent believes the opponent is equally likely to cooperate or defect ($q_i = 1/2$), the belief-dependent cutoff collapses exactly to $\alpha_{\mathrm{cut}}(1/2) = \bar{\alpha}$. Thus, the incomplete-information framework of Section~\ref{sec:bayesian} does not discard the complete-information analysis, but rather generalizes its single knife-edge threshold into a continuous family of thresholds indexed by the player's subjective belief $q_i$.

\begin{lemma}[Cutoff Strategies]
\label{lem:cutoff}
Player $i$'s best response is a cutoff rule: $f_i(\alpha_i)=1$ (cooperate) if $\alpha_i\ge \alpha^{\mathrm{cut}}(q_i)$ and $f_i(\alpha_i)=0$ otherwise, where
\[
\alpha^{\mathrm{cut}}(q) := \frac{(P-S)+qK}{(T-P)-qK}, \qquad \alpha^{\mathrm{cut}}(0)=\frac{P-S}{T-P},\quad \alpha^{\mathrm{cut}}(1)=\frac{T-R}{R-S},\quad \alpha^{\mathrm{cut}}\!\left(\tfrac12\right)=\overline{\alpha},
\]
with $\overline{\alpha}=\dfrac{T-R+P-S}{T+R-P-S}$ as in Section~\ref{subsec:altruism-effects}.
\end{lemma}

\begin{proof}
    See Appendix \ref{app:proof-cutoff_1}.
\end{proof}

The identity $\alpha^{\mathrm{cut}}(1/2)=\overline{\alpha}$ gives a clean bridge to the complete-information analysis: $\overline{\alpha}$ is exactly the cutoff a player would use if she believed the opponent cooperates with probability one-half.
\begin{lemma}
\label{lem:cutoff-monotonicity}
\begin{align*}
    \dfrac{d\,\alpha^{\mathrm{cut}}}{dq} = \dfrac{(T-S)K}{\big((T-P)-qK\big)^2}, 
\end{align*}
whose sign equals the sign of $K :=T-R-(P-S)$.
\end{lemma}

\begin{proof}
    See Appendix~\ref{app:proof-cutoff_2}.
\end{proof}

When $T-R\ge P-S$, believing the opponent is more likely to cooperate \emph{raises} one's own cutoff: the temptation to free-ride against an (expected) cooperator is strong enough that greater altruism is needed to resist it. When $T-R\le P-S$, the opposite holds: believing the opponent cooperates \emph{lowers} the altruism required to cooperate oneself.

\subsubsection{Existence and (non-)uniqueness of equilibrium}

\begin{proposition}[Bayesian Cutoff Equilibrium]
\label{prop:cutoff-equilibrium}
For any pair of (possibly heterogeneous) beliefs $\mu_i,\mu_j\in\Delta(\mathcal{A})$, there exists a Bayesian Nash equilibrium in cutoff strategies: $f_i(\alpha_i)=\mathbbm{1}[\alpha_i\ge\alpha_i^{\mathrm{cut}}]$, $f_j(\alpha_j)=\mathbbm{1}[\alpha_j\ge\alpha_j^{\mathrm{cut}}]$, where $(q_i^*,q_j^*)\in[0,1]^2$ solves
\begin{align*}
    q_i^* = \mu_i\big(\alpha_j\ge \alpha^{\mathrm{cut}}(q_j^*)\big), \qquad q_j^* = \mu_j\big(\alpha_i\ge \alpha^{\mathrm{cut}}(q_i^*)\big),
\end{align*}
and $\alpha_i^{\mathrm{cut}} := \alpha^{\mathrm{cut}}(q_j^*)$, $\alpha_j^{\mathrm{cut}} := \alpha^{\mathrm{cut}}(q_i^*)$.
\end{proposition}

\begin{proof}
See Appendix~\ref{app:proof-cutoff-equilibirum}.
\end{proof}

Under complete information, an agent's cooperation decision reduces to a one-directional threshold comparison against the opponent's known altruism. Once altruism is private information, however, a player's optimal cutoff depends on her belief about the opponent's action, which itself must be consistent with the opponent's belief about her. This interdependence replaces a unilateral decision rule with a fixed-point problem: in the regime where fear dominates temptation ($K < 0$), it introduces strategic complementarity and opens the door to multiple, self-fulfilling belief-driven equilibria.

\begin{proposition}[Symmetric Equilibrium: Uniqueness versus Multiplicity]
\label{prop:symmetric-equilibrium}
Suppose $\mu_i=\mu_j=\overline{\mu}$ (common prior, support $\mathcal{A}=[\min\{\alpha^{\mathrm{cut}}(0),\alpha^{\mathrm{cut}}(1)\},\max\{\alpha^{\mathrm{cut}}(0),\alpha^{\mathrm{cut}}(1)\}]$) and let $q^*:=q_i^*=q_j^*$ solve $q=\overline{\mu}\big(\alpha\ge\alpha^{\mathrm{cut}}(q)\big)=:\mathcal{R}(q)$.
\begin{enumerate}
    \item If $T-R\ge P-S$ ($K\ge0$), $\alpha^{\mathrm{cut}}$ is nondecreasing in $q$ (Lemma~\ref{lem:cutoff-monotonicity}), so $\mathcal{R}$ is nonincreasing; hence $g(q):=\mathcal{R}(q)-q$ is strictly decreasing, and the symmetric equilibrium $q^*$ is \emph{unique}.

    \item If $T-R\le P-S$ ($K\le0$), $\alpha^{\mathrm{cut}}$ is nonincreasing in $q$, so $\mathcal{R}$ is nondecreasing: beliefs about the opponent's cooperation are strategic complements, and \emph{multiple} symmetric equilibria can coexist.
\end{enumerate}
\end{proposition}

\begin{proof}
See Appendix~\ref{app:proof-symmetric}.
\end{proof}

Case~2 of Proposition~\ref{prop:symmetric-equilibrium} is the paper's central positive result on multiplicity: private information about altruism transforms cooperation from a disposition problem into a \emph{coordination} problem whenever fear/risk dominates greed/temptation in the payoff structure. Nothing about the payoffs or the (correctly specified, common) prior distinguishes the two equilibria in this regime --- only which belief players happen to coordinate on.

\begin{figure}[h]
\centering
\includegraphics[width=1.0\textwidth]{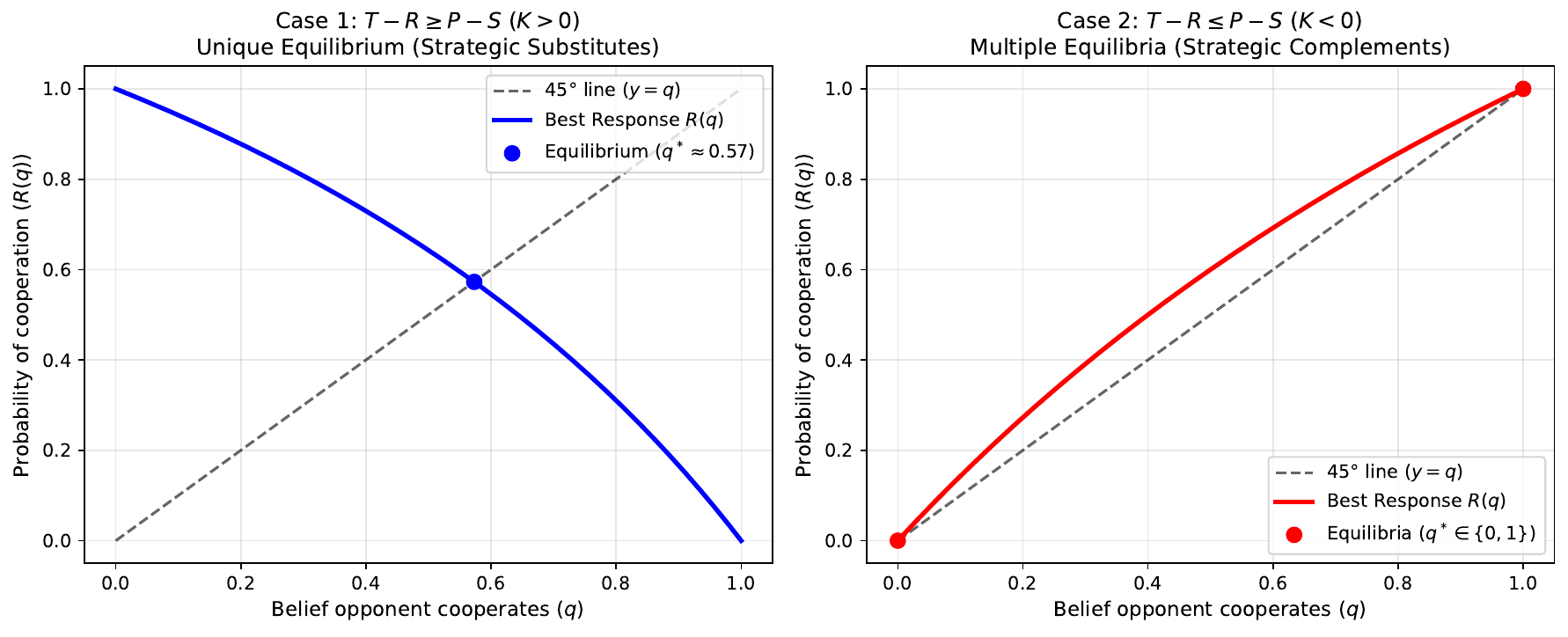}
\caption{Symmetric Bayesian Nash Equilibrium (Proposition 4.2). The figure plots the aggregate best-response function $R(q) = \mu(\alpha \ge \alpha_{cut}(q))$ against the subjective belief $q$ that the opponent cooperates, with a uniform prior $\mu$ over the type support. \textbf{Left panel ($K>0$)}: When greed/temptation dominates, actions are strategic substitutes, yielding a unique interior equilibrium. \textbf{Right panel ($K<0$)}: When fear/risk dominates, actions are strategic complements, yielding multiple self-fulfilling equilibria (universal defection and universal cooperation). This conceptual graph was created using Gemini 3.1 Pro.The author reviewed and edited the content as needed and takes full responsibility for the content of the graph.}
\label{fig:cutoff-br}
\end{figure}

For a numerical illustration of case~1, take $(R,P,S,T)=(5,1,0,10)$ (so $K=4>0$, $\overline{\alpha}=0.4286$, $\mathcal{A}=[0.111,1.0]$) with $\overline{\mu}=\mathrm{Unif}(\mathcal{A})$. The unique symmetric equilibrium is $q^*\approx0.573$, $\alpha^{\mathrm{cut}}(q^*)\approx0.491$; naive best-response iteration does not converge monotonically to it (it 2-cycles between $q\approx0.50$ and $q\approx0.64$), underscoring that $q^*$ must be located by solving the fixed-point equation directly rather than by simulating best responses.

For case~2, take $(R,P,S,T)=(1,0,-2,2)$ (so $K=-1<0$, $\mathcal{A}=[1/3,1]$) with $\overline\mu=\mathrm{Unif}(\mathcal{A})$. Direct computation gives $\mathcal{R}(0)=0$ and $\mathcal{R}(1)=1$ exactly, with $\mathcal{R}(q)>q$ strictly on $(0,1)$: there are exactly two symmetric equilibria, $q^*=0$ (a self-fulfilling universal-defection equilibrium: believing no one cooperates raises everyone's cutoff to the top of the support) and $q^*=1$ (a self-fulfilling universal-cooperation equilibrium: believing everyone cooperates lowers everyone's cutoff to the bottom of the support), for the \emph{same} payoff structure and the \emph{same}, correctly specified, common prior. We return to this multiplicity in Section~\ref{subsec:uncertainty-cooperation}.

\begin{remark}
The characterization in this subsection supersedes a simpler, reduced-form alternative in which one evaluates the \emph{complete-information} formula $\rho_i^*(C;\alpha_j)$ of Proposition~\ref{prop:nash-efficiency} directly against player $i$'s belief $\mu_i$ over the realized $\alpha_j$, declaring cooperation ``typical'' whenever $\mu_i(\alpha_j\ge\overline{\alpha})\ge \mu_i(\alpha_j<\overline{\alpha})$. That shortcut --- used again in Section~\ref{subsec:robust-bne} below for tractability --- does not take the opponent's own equilibrium response into account and should be understood as an approximation to Proposition~\ref{prop:cutoff-equilibrium}, exact only at the symmetric benchmark $q=1/2$.
\end{remark}

\subsection{Effects of Payoff Structures}
\label{subsec:bayesian-payoff-effects}
Let $K :=T-R-(P-S)$. We obtain the following result.
\begin{proposition}
\label{prop:cutoff-payoff-comparative-statics}
For every $q\in[0,1]$,
\begin{align*}
    \frac{\partial \alpha^{\mathrm{cut}}}{\partial R} = \frac{q(S-T)}{\big((T-P)-qK\big)^2} \le 0, \qquad \frac{\partial \alpha^{\mathrm{cut}}}{\partial P} = \frac{(S-T)(q-1)}{\big((T-P)-qK\big)^2} \ge 0,
\end{align*}
regardless of the sign of $K$. By contrast, $\partial\alpha^{\mathrm{cut}}/\partial S$ and $\partial\alpha^{\mathrm{cut}}/\partial T$ do not have an unconditional sign; at the symmetric benchmark $q=1/2$ (where $\alpha^{\mathrm{cut}}=\overline\alpha$),
\begin{align*}
    \frac{\partial \overline\alpha}{\partial R} = \frac{2(S-T)}{(T+R-P-S)^2} < 0,\quad \frac{\partial \overline\alpha}{\partial P} = \frac{2(T-S)}{(T+R-P-S)^2} > 0,
\end{align*}
\begin{align*}
    \frac{\partial \overline\alpha}{\partial S} = \frac{2(P-R)}{(T+R-P-S)^2} < 0, \quad \frac{\partial \overline\alpha}{\partial T} = \frac{2(R-P)}{(T+R-P-S)^2} > 0.
\end{align*}
\end{proposition}

\begin{proof}
See Appendix~\ref{app:proof-cutoff-payoff}.
\end{proof}

The belief-based cutoff needed to sustain cooperation falls as $R$ rises or $P$ falls --- and, unlike the complete-information comparative statics of Corollary~\ref{cor:payoff-comparative-statics}, this holds \emph{unconditionally}, for every belief level $q$ and independent of which of $T-R$ or $P-S$ dominates. The effects of $S$ and $T$ retain the sign pattern of \citeauthor{AOSSW_2001}'s (\citeyear{AOSSW_2001}) fear/greed indices at the symmetric benchmark $q=1/2$, but are level-dependent away from it: the more a player already trusts her opponent to cooperate, the less (respectively more) sensitive her own threshold becomes to the temptation and sucker payoffs. Together with Section~\ref{subsec:payoff-effects}, this gives a \emph{belief-based cooperation index} --- the threshold $\alpha^{\mathrm{cut}}(q)$ --- that runs parallel to the payoff-based cooperation index of \citet{AOSSW_2001}.

\begin{figure}[h]
\centering
\includegraphics[width=0.55\textwidth]{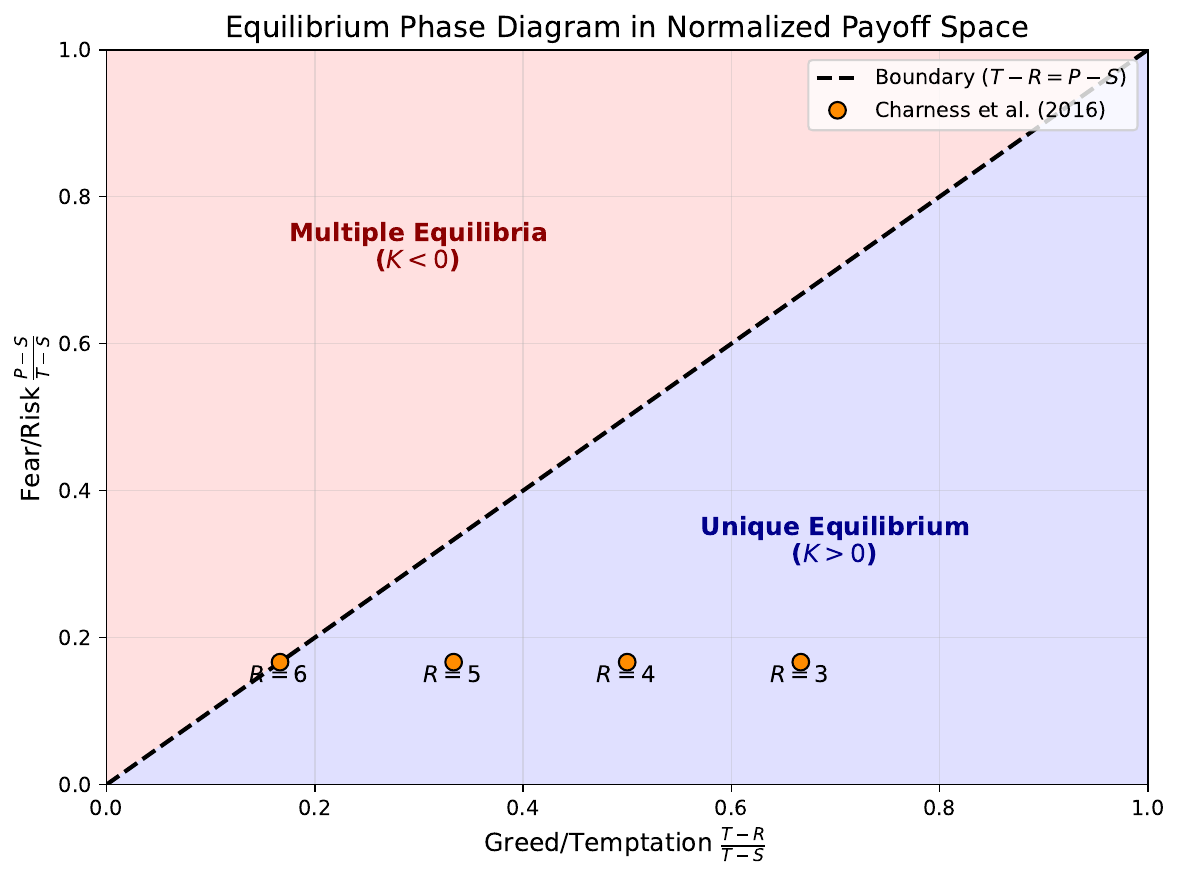}
\caption{Equilibrium Phase Diagram in the Normalized payoff space following \citet{AOSSW_2001}, with experimental data points from \citet{CRR_2016}. The blue region ($K>0$) yields a unique symmetric equilibrium, while the red region ($K<0$) allows for multiple equilibria. The orange markers plot the experimental treatments of \citet{CRR_2016} with $P=2, T=7, S=1$ and $R \in \{3,4,5,6\}$. The highest-cooperation treatment ($R=6$) sits exactly on the boundary $K=0$, illustrating that the multiplicity regime is not a theoretical curiosity but lies adjacent to standard experimental designs. This conceptual graph was created using Gemini 3.1 Pro.The author reviewed and edited the content as needed and takes full responsibility for the content of the graph.}
\label{fig:phase-diagram}
\end{figure}

Section~\ref{subsec:belief-effects} showed that beliefs alone, without any change in the underlying preferences or payoffs, can be enough to select between a cooperative and an uncooperative equilibrium (Proposition~\ref{prop:symmetric-equilibrium}). This raises a question the subjective Bayesian equilibrium of Section~\ref{sec:bayesian} is not equipped to answer: \emph{which} belief should a player actually hold, and how much should she trust it? A subjective prior $\mu_i$ is, after all, still just a guess --- typically formed from a small sample of past interactions, secondhand reports, or introspection --- and a player who is aware that her guess could be wrong may not want to act as if it were exactly right. We now ask what happens to the cutoff equilibria of Section~\ref{sec:bayesian} when players explicitly guard against this kind of model misspecification, rather than treating their reference belief $\overline\mu$ as beyond doubt.

We emphasize that Figure~\ref{fig:phase-diagram} serves as a structural illustration of parameter proximity rather than a direct empirical test of multiplicity. To cleanly identify multiple equilibria from preference heterogeneity, future experimental work could examine whether cooperation rates across independent matching groups under identical payoff parameters exhibit bimodal clustering (coordination on high vs. low cooperation regimes) rather than a unimodal distribution around a single intermediate rate, especially when exogenous coordinating devices (such as brief pre-play communication or framing) are varied.

\section{Robustness and Cooperation}
\label{sec:robust}

Sections~\ref{sec:nash} and~\ref{sec:bayesian} together establish a three-tier epistemic hierarchy behind cooperation in $\mathcal{G}_{PD}$. Section~\ref{sec:nash} shows that cooperation requires sufficient \emph{disposition}: the opponent's altruism $\alpha_j$ must clear a threshold. Section~\ref{sec:bayesian} shows that disposition alone is neither necessary nor sufficient once altruism is private information: what matters is a player's \emph{belief} $\mu_i$ about $\alpha_j$, and beliefs alone --- with actual altruism held fixed, even at a low average level --- can determine whether cooperation is sustained. This section adds a third tier. A belief is itself only a guess, and a player who does not fully trust that guess will act on a \emph{robustified} version of it, $\mu_i^*\neq\mu_i$, rather than on the guess itself. We show that this robustification is a force in its own right: even when the reference belief $\overline\mu$ correctly describes an opponent population that supports cooperation, insufficient confidence in that belief --- captured by the robustness parameter $\theta_i$ --- can undo it. Disposition, belief, and robustness are each individually necessary but not sufficient for cooperation: social preferences alone do not generate cooperation.

\subsection{Robust Beliefs and Cooperation}
\label{subsec:robust-bne}
Subjective beliefs (Definition~\ref{def:subjective-bne}) need not be well calibrated. We now explore how a player who is \emph{averse to model misspecification} --- who does not fully trust her own best-guess belief $\overline{\mu}$ about the opponent's altruism --- will distort that belief, and whether cooperation survives this distortion.

Following \citeauthor{HS_2001}'s (\citeyear{HS_2001}) multiplier-preferences approach to robust control, let $H(\mu_i)=-\int_{\mathcal{A}}\mu_i(\alpha)\log\mu_i(\alpha)\,d\mu_i$ denote entropy, and let $R(\mu_i\Vert\overline{\mu})=\int\log\!\big(\mu_i(\alpha)/\overline{\mu}(\alpha)\big)\,d\mu_i$ denote the relative entropy of $\mu_i$ with respect to a reference (``best guess'') prior $\overline{\mu}$, whenever $\mu_i\ll\overline{\mu}$ (and $+\infty$ otherwise). Given a robustness parameter $\theta_i\in(0,\infty]$, the multiplier-preferences representation of player $i$'s subjective expected payoff solves
\begin{align*}
\phi_{\theta_i}^{-1}\big(\mathbb{E}_{\mu_i}[\widehat{g}_i(\cdot)]\big) = \min_{\mu_i\in\Delta(\mathcal{A})}\left(\int_{\mathcal{A}} \widehat{g}_i(\cdot)\,d\mu_i(\alpha_j\mid\alpha_i) + \theta_i R(\mu_i\Vert\overline{\mu})\right).
\end{align*}

\begin{definition}[Subjective BNE with Multiplier Preferences]
\label{def:multiplier-bne}
A strategy profile $f^*$ is a \emph{subjective Bayesian Nash equilibrium with multiplier preferences} if for all $i\in I$, $\alpha_i\in\mathcal{A}_i$, and $f_i'$, there exists
\begin{align*}
    \mu_i^*\in \operatorname*{argmin}_{\mu_i\in\Delta(\mathcal{A}_j)}\left(\int_{\mathcal{A}_j}\widehat{g}_i(f_i(\alpha_i),f(\alpha_j),\alpha)\,d\mu_i(\alpha_j\mid\alpha_i) + \theta_i R(\mu_i\Vert\overline{\mu})\right)
\end{align*}
such that $\mathbb{E}_{\mu_i^*}[f^*(\alpha)\mid \alpha_i,\overline{\mu}] \ge \mathbb{E}_{\mu_i^*}[f_i(\alpha_i),f_j^*(\alpha_j)\mid\alpha_i,\overline{\mu}]$.
\end{definition}

\begin{remark}[Robust Control vs.\ Full Bayesian Learning]
\label{rem:robust-vs-bayesian}
A standard alternative to robust control is a hierarchical Bayesian model in which the player specifies a prior over possible type distributions (a hyper-prior). However, in one-shot or unfamiliar interactions, an agent rarely possesses the objective data necessary to discipline such second-order beliefs. Multiplier preferences \citep{HS_2001, S_2011} directly capture an agent's concern that her reference model $\bar{\mu}_i$ is misspecified, without introducing unobservable hyper-priors. Moreover, as established in Lemma~\ref{lem:tilting}, this formulation yields an analytically tractable exponential tilting that isolates a single, interpretable epistemic confidence threshold $\theta_i^\dagger$, providing a clean bridge between belief distortion and cooperation viability.
\end{remark}

To make the minimization explicit, define the linear \emph{cooperation-favorability statistic}
\begin{align*}
    L(\alpha_j) := (T+R-P-S)\,\alpha_j - (T-R+P-S),
\end{align*}
whose sign coincides with whether $\rho_i^*(C;\alpha_j)\ge \tfrac12$ (Section~\ref{subsec:altruism-effects}): $L(\overline\alpha)=0$ and $L$ is increasing in $\alpha_j$ since $T+R>P+S$.

\begin{lemma}[Exponential Tilting]
\label{lem:tilting}
The minimizer in Definition~\ref{def:multiplier-bne}, applied with $L$ in place of the payoff integrand, has the closed form
\begin{align*}
    \mu_i^*(\alpha_j) = \frac{\overline{\mu}(\alpha_j)\exp\!\big(-L(\alpha_j)/\theta_i\big)}{\displaystyle\int_{\mathcal{A}_j}\overline{\mu}(\alpha_j')\exp\!\big(-L(\alpha_j')/\theta_i\big)\,d\alpha_j'}.
\end{align*}
\end{lemma}

\begin{proof}
See Appendix~\ref{app:proof-tilting} for the full Donsker--Varadhan-style derivation. Since $L$ is affine in $\alpha_j$, $\mu_i^*$ is an exponential tilting of $\overline{\mu}$ toward lower values of $\alpha_j$ --- a pessimistic, worst-case belief about the opponent's altruism.
\end{proof}

Let $\mathcal{M}_i := \{\mu_i\in\Delta(\mathcal{A})\mid \mu_i(\alpha_j\ge\overline\alpha)\ge\mu_i(\alpha_j<\overline\alpha)\}$ denote the set of beliefs that place at least as much mass above the crossing point $\overline\alpha$ as below it.

\begin{proposition}[Cooperation under Model Uncertainty]
\label{prop:cooperation-model-uncertainty}
Given $\mathcal{G}_{PD}$ with $T-R\ge P-S$, suppose each $i\in I$ has efficiency-concern preferences with multiplier preferences $(\theta_i,\overline{\mu}_i)$, and let $\mu_i^*$ be the resulting optimal belief. Then, for each $i\in I$, if $\mu_i^*\in\mathcal{M}_i$, then $\rho_i(C,\{C,D\})\ge\rho_i(D,\{C,D\})$.
\end{proposition}

\begin{proof}
See Appendix~\ref{app:proof-cooperation-uncertainty}.
\end{proof}

Proposition~\ref{prop:cooperation-model-uncertainty} is stated conditionally on $\mu_i^*\in\mathcal{M}_i$. Lemma~\ref{lem:tilting} lets us characterize exactly when this condition holds.

Proposition~\ref{prop:robustness-threshold} can be viewed as the robust counterpart of the cutoff characterization in Section~\ref{sec:bayesian}. Under model uncertainty, a player's cooperation condition remains $\alpha_i \ge \alpha_{\mathrm{cut}}(\tilde{q}_i(\theta_i))$, where the baseline belief $q_i$ is replaced by the pessimistically distorted belief $\tilde{q}_i(\theta_i)$ resulting from exponential tilting. The threshold $\theta_i^\dagger$ is therefore the minimum epistemic confidence required for the subjective cooperation condition to remain satisfied under robustness:
\[
\theta_i \ge \theta_i^\dagger \iff \tilde{q}_i(\theta_i) \ge q_i^* \iff \alpha_i \ge \alpha_{\mathrm{cut}}(\tilde{q}_i(\theta_i)).
\]

\begin{proposition}[Robustness and the Viability of Cooperation]
\label{prop:robustness-threshold}
Suppose $\overline{\mu}(\alpha_j\ge\overline\alpha)>0$. Then $\mu_i^*(\alpha_j\ge\overline\alpha)$ is monotonically increasing in $\theta_i$, converging to $\overline{\mu}(\alpha_j\ge\overline\alpha)$ as $\theta_i\to\infty$ and to $0$ as $\theta_i\to0$. Consequently, there exists a unique threshold $\theta_i^\dagger\in(0,\infty)$ such that
\[
\mu_i^*\in\mathcal{M}_i \iff \theta_i\ge\theta_i^\dagger.
\]
\end{proposition}

\begin{proof}
See Appendix~\ref{app:proof-robustness-threshold}.
\end{proof}

Proposition~\ref{prop:robustness-threshold} replaces the hypothesis of Proposition~\ref{prop:cooperation-model-uncertainty} with a primitive condition on the model's deep robustness parameter $\theta_i$. For example, with $\overline\mu=\mathrm{Unif}(\mathcal{A})$ and $(R,P,S,T)=(5,1,0,10)$ ($\overline\alpha=0.4286$, $\mathcal{A}=[0.111,1.0]$), numerically $\mu_i^*(\alpha_j\ge\overline\alpha)$ equals $0.012$, $0.358$, $0.570$ at $\theta_i=1,5,20$ respectively, approaching the uniform baseline of $0.643$ as $\theta_i\to\infty$, with $\theta_i^\dagger\approx10.28$ (Figure~\ref{fig:altruism-trust}). Even though the reference model $\overline\mu$ itself supports cooperation ($0.643>0.5$), sufficiently strong concern for model misspecification (small $\theta_i$) can unravel it: pessimism about the opponent's type is, formally, indistinguishable from genuine pessimism about the opponent's altruism.

\begin{figure}[h]
\centering
\includegraphics[width=0.55\textwidth]{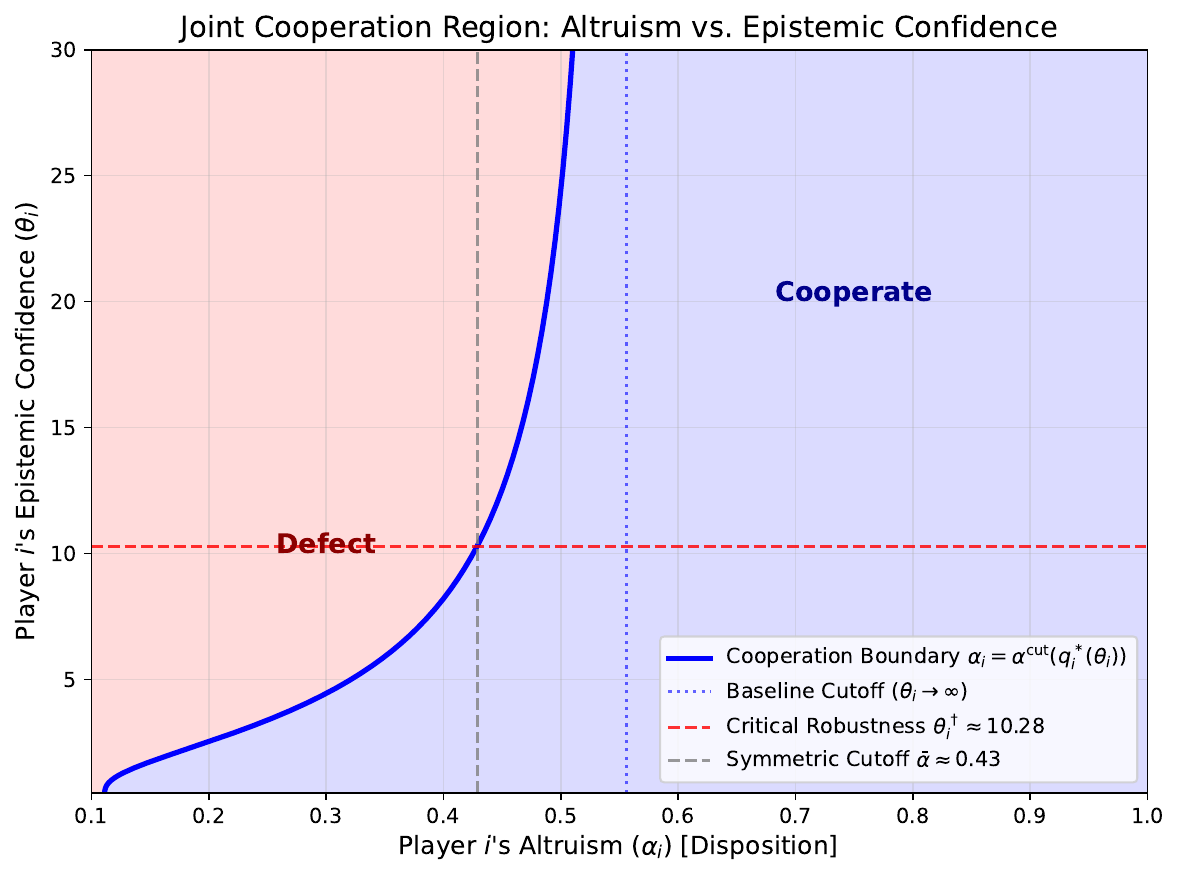}
\caption{Joint Cooperation Region (Altruism vs. Epistemic Confidence). The figure illustrates the boundary of cooperation in the parameter space of player $i$'s disposition (altruism $\alpha_i$) and epistemic attitude (robustness $\theta_i$), using the payoff structure $(R,P,S,T) = (5,1,0,10)$ and a uniform reference prior $\mu$. A player cooperates if and only if she falls into the blue region. Even with sufficient altruism ($\alpha_i > \overline{\alpha}$), cooperation collapses if the player's epistemic confidence in her reference prior falls below the critical threshold $\theta_i^\dagger \approx 10.28$. Note that the two axes are not commensurable: $\alpha_i$ is a preference parameter, while $\theta_i$ is an epistemic one. This conceptual graph was created using Gemini 3.1 Pro.The author reviewed and edited the content as needed and takes full responsibility for the content of the graph.}
\label{fig:altruism-trust}
\end{figure}

\paragraph{What determines $\theta_i^\dagger$?} The threshold $\theta_i^\dagger$ is itself governed by the payoff structure. Figure~\ref{fig:theta-dagger-vs-R} plots $\theta_i^\dagger$ against $R$, holding $(P,S,T)=(1,0,10)$ fixed and restricting to the Pareto-efficient region $2R>T+S$, i.e., $R>5$: the robustness required to sustain cooperation \emph{falls} monotonically as the cooperative surplus $R$ rises, from $\theta_i^\dagger\approx10.27$ near $R=5$ to $\theta_i^\dagger\approx10.00$ at $R=8.95$. This mirrors the belief-cutoff comparative static $\partial\overline\alpha/\partial R<0$ of Proposition~\ref{prop:cutoff-payoff-comparative-statics}: a more efficient cooperative outcome lowers both the altruism and the confidence a player needs before she is willing to cooperate.

\begin{figure}[h]
\centering
\includegraphics[width=0.62\textwidth]{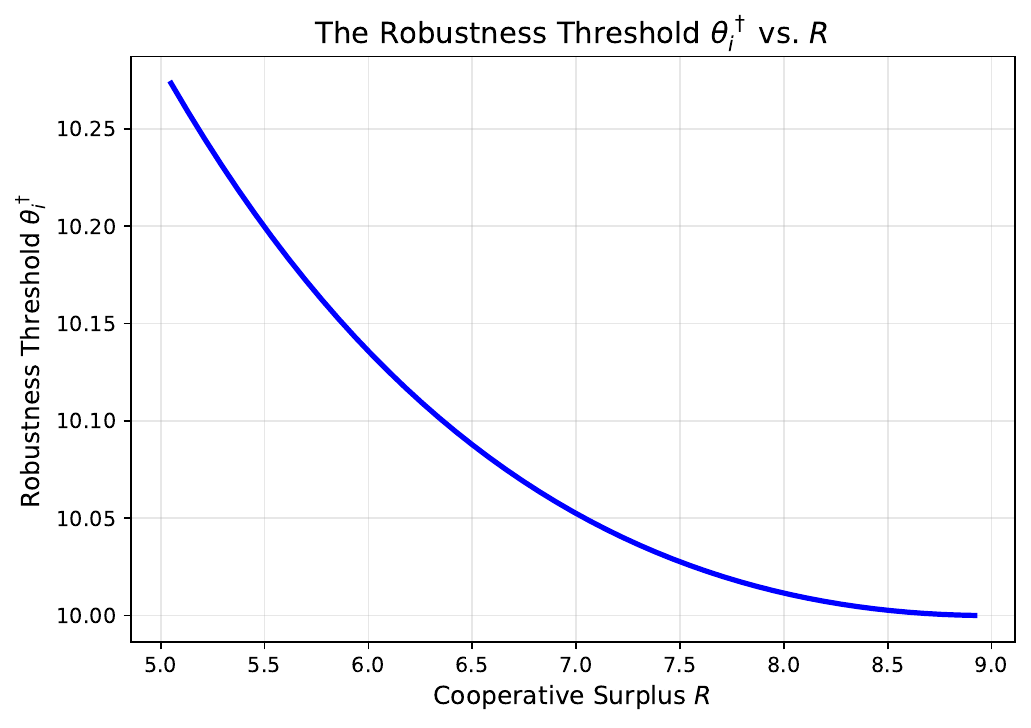}
\caption{The robustness threshold $\theta_i^\dagger$ as a function of $R$, holding $(P,S,T)=(1,0,10)$ fixed, restricted to the Pareto-efficient region $R>5$. Computed numerically via bisection on $\mu_i^*(\alpha_j\ge\overline\alpha;\theta_i)=1/2$. This conceptual graph was created using Gemini 3.1 Pro. The author reviewed and edited the content as needed and takes full responsibility for the content of the graph.}
\label{fig:theta-dagger-vs-R}
\end{figure}

\subsubsection*{Three interpretations}

\paragraph{Altruism versus trust.} 
We interpret $\theta_i$ as reflecting epistemic confidence in one's reference prior. The parameters $\alpha_i$ and $\theta_i$ measure two conceptually distinct things. $\alpha_i$ is a disposition: how much player $i$ herself is willing to sacrifice for the opponent's payoff. $\theta_i$ is an epistemic attitude: how much player $i$ trusts her own belief about the opponent's disposition. A population can vary along either dimension independently of the other, and Proposition~\ref{prop:robustness-threshold} shows that both must clear a bar for cooperation to survive. This separates two behaviorally distinct channels --- generosity, and confidence in others' generosity --- that a single cooperation regression on $(R,P,S,T)$ cannot distinguish, but that could in principle be measured separately, e.g., via an elicited social value orientation for $\alpha_i$ and an elicited measure of confidence or ambiguity aversion over beliefs for $\theta_i$.

\paragraph{Belief-driven cooperation is fragile.} Under complete information (Section~\ref{sec:nash}), a sufficiently altruistic population cooperates, full stop. Under incomplete information (Section~\ref{sec:bayesian}), a sufficiently optimistic \emph{belief} about the population's altruism suffices, even when actual altruism is weak. Proposition~\ref{prop:robustness-threshold} shows that this second, belief-driven route to cooperation is fragile in a way the first is not: it survives only above the robustness threshold $\theta_i^\dagger$, and can be undone purely by the player's own doubt about her belief, with the underlying population of altruism types held completely fixed. Cooperation sustained by beliefs about altruism is thus less stable than cooperation sustained by altruism itself.

\paragraph{Cooperation without highly altruistic opponents.} Read together with Section~\ref{sec:bayesian}, this section's results imply that cooperation does not require the opponent to be altruistic on average. What it requires is (i) a belief that assigns enough mass to sufficiently altruistic types, and (ii) enough confidence in that belief, $\theta_i\ge\theta_i^\dagger$, to act on it. Even a correctly specified reference model $\overline\mu$ that itself supports cooperation can fail to sustain it once (ii) is violated, so the binding constraint on cooperation is often not how altruistic the opponent actually is, but how robustly a player is willing to believe it.

\subsection{Robust Logit Responses}
\label{subsec:robust-logit-qre}
The results above are this paper's main claims about \emph{when cooperation survives}. The remainder of this section is a formal extension rather than a further behavioral claim: it shows that the same multiplier-preferences machinery, applied to a player's own \emph{action} choice instead of her belief, nests our framework inside the quantal response literature and clarifies exactly how Sections~\ref{sec:nash}--\ref{sec:robust} relate to one another as limiting cases of a single equilibrium concept. A reader interested only in the cooperation results above can skip to Section~\ref{sec:discussion} without loss of continuity.

Multiplier preferences discipline the \emph{belief} about the opponent's type. We now apply the same robust-control logic to the player's own \emph{action} choice, which delivers a smooth (quantal) response rather than a sharp best response.

\begin{lemma}[Entropy Regularization and Logit Response]
\label{lem:logit}
Suppose player $i$'s mixed strategy $\rho_i\in\Delta(\{C,D\})$ itself solves a multiplier-preferences problem with reference distribution the uniform distribution on $\{C,D\}$ and precision parameter $\lambda_i>0$:
\begin{align*}
    \rho_i \in \operatorname*{argmax}_{\rho \in \Delta(\{C,D\})}\Big[\rho\,u(C) + (1-\rho)\,u(D) - \lambda_i \big(\rho\log\rho + (1-\rho)\log(1-\rho)\big)\Big].
\end{align*}
Then, 
\begin{align*}
    \rho_i(C) = \frac{\exp(u(C)/\lambda_i)}{\exp(u(C)/\lambda_i)+\exp(u(D)/\lambda_i)}.
\end{align*}
\end{lemma}

\begin{proof}
Differentiate the objective with respect to $\rho$ and set the derivative to zero: $u(C)-u(D)=\lambda_i\log\big(\rho/(1-\rho)\big)$, which rearranges to the stated logit expression.
\end{proof}

This is precisely the response function of the logit Quantal Response Equilibrium (QRE) of \citet{MP_1995}: rather than best-responding sharply, player $i$ chooses actions with probability increasing in expected payoff, with $\lambda_i$ governing the degree of choice noise ($\lambda_i\to0$ recovers a sharp best response; $\lambda_i\to\infty$ recovers uniform randomization).

\begin{definition}[Robust Logit Quantal Response Equilibrium]
\label{def:robust-logit-qre}
Given parameters $(\theta_i,\lambda_i,\overline\mu_i)_{i\in I}$, a profile $(\rho_i^*,\mu_i^*)_{i\in I}$ is a \emph{Robust Logit QRE} if, for each $i\in I$:
\begin{enumerate}
\item[(a)] $\mu_i^*$ is the robust-optimal belief of Lemma~\ref{lem:tilting}, with $L$ replaced by the expected payoff difference induced by the opponent's equilibrium logit response $\rho_j^*(\cdot)$; and
\item[(b)] $\rho_i^*(C) = \dfrac{\exp\big(\mathbb{E}_{\mu_i^*}[g_i(C,\cdot)]/\lambda_i\big)}{\exp\big(\mathbb{E}_{\mu_i^*}[g_i(C,\cdot)]/\lambda_i\big)+\exp\big(\mathbb{E}_{\mu_i^*}[g_i(D,\cdot)]/\lambda_i\big)}$.
\end{enumerate}
\end{definition}

The two parameters play distinct roles: $\theta_i$ governs \emph{ambiguity aversion} about the opponent's type (Section~\ref{subsec:robust-bne}), while $\lambda_i$ governs \emph{choice noise} in the player's own action (Lemma~\ref{lem:logit}). Robust Logit QRE nests both channels simultaneously.

\paragraph{Relation to logit QRE.}

\begin{proposition}
\label{prop:relation-qre}
Fix $(\lambda_i)_{i\in I}$. As $\theta_i\to\infty$ for all $i\in I$, Robust Logit QRE converges to the standard logit QRE of \citet{MP_1995}, computed under the common reference prior $\overline\mu$.
\end{proposition}

\begin{proof}
By Lemma~\ref{lem:tilting}, the tilting factor $\exp(-L(\alpha_j)/\theta_i)\to1$ as $\theta_i\to\infty$, so $\mu_i^*\to\overline\mu_i$. Substituting $\overline\mu_i$ for $\mu_i^*$ in Definition~\ref{def:robust-logit-qre}(b) recovers the ordinary logit QRE response.
\end{proof}

\paragraph{Relation to Nash equilibrium.}
\label{subsec:relation-ne}

\begin{proposition}
\label{prop:relation-ne}
As $\theta_i\to\infty$ and $\lambda_i\to0$ for all $i\in I$, Robust Logit QRE converges to the Bayesian Nash equilibrium of Section~\ref{sec:bayesian}; if beliefs additionally degenerate to complete information, it converges to the Nash equilibrium of Proposition~\ref{prop:nash-efficiency}.
\end{proposition}

\begin{proof}
By Proposition~\ref{prop:relation-qre}, $\theta_i\to\infty$ removes the belief distortion. As $\lambda_i\to0$, the logit response of Lemma~\ref{lem:logit} converges to a sharp best response: $\rho_i(C)\to\mathbb{1}\big[\mathbb{E}[g_i(C,\cdot)]>\mathbb{E}[g_i(D,\cdot)]\big]$ (indeterminate at exact ties), which is the standard logit-QRE-to-Nash limit \citep{MP_1995} applied to our Bayesian game.
\end{proof}

Propositions~\ref{prop:relation-qre} and~\ref{prop:relation-ne} together establish the nested hierarchy
\[
\text{Nash Equilibrium} \ \subset\ \text{Bayesian Nash Equilibrium} \ \subset\ \text{Logit QRE} \ \subset\ \text{Robust Logit QRE},
\]
obtained as the two-parameter limit $(\theta_i,\lambda_i)\to(\infty,0)$: robust control (via $\theta_i$) and quantal response (via $\lambda_i$) are independent generalizations of Nash equilibrium, and Robust Logit QRE is their common generalization. We emphasize that this hierarchy is a formal packaging device, not an additional behavioral claim: the cooperation results of Sections~\ref{subsec:robust-bne}--\ref{subsec:uncertainty-cooperation} above hold already at $\lambda_i=0$ and finite $\theta_i$, and do not depend on the logit-response extension developed in this subsection.

\paragraph{Axiomatic Foundations and Non-Expected Utility.}
Our formulation of the Robust Logit QRE can be viewed as a concrete, applied realization of recent axiomatic advancements in non-expected utility games. \cite{MPST_2024} establish that any descriptive statistic satisfying monotonicity and additivity for independent lotteries must be a mixture of CARA certainty equivalents. This directly corresponds to the exponential tilting mechanism (multiplier preferences) we employ to capture the epistemic confidence $\theta_i$. 

Furthermore, \cite{SSTW_2025} demonstrate that if players satisfy a ``bracketing'' axiom---treating independent games separately---and exhibit distribution-monotonicity, their equilibrium behavior must be a Statistic Response Equilibrium (SRE) or its logit variant. In this light, our Robust Logit QRE is exactly the logit SRE evaluated under the specific monotone additive statistic dictated by robust control. While \cite{SSTW_2025} provide the top-down axiomatic justification for why such non-expected utility responses are fundamentally consistent in games, our model provides the bottom-up, dynamic implications: we show exactly how this robust-control mechanism interacts with social preferences to unravel cooperation when the epistemic confidence falls below the critical threshold $\theta_i^\dagger$.

\section{Discussion}
\label{sec:discussion}

\subsection{Payoff Structures and Belief Structures}
\label{subsec:discussion-payoff-belief}
Sections~\ref{sec:nash} and~\ref{sec:bayesian} offer two cooperation indices built on the same payoff primitives $(R,P,S,T)$: the complete-information index of Corollary~\ref{cor:payoff-comparative-statics}, and the belief-based cutoff $\alpha^{\mathrm{cut}}(q)$ of Proposition~\ref{prop:cutoff-payoff-comparative-statics}. A notable difference is robustness of sign: under complete information, the comparative statics in $R$ and $P$ depend on whether greed/temptation or fear/risk dominates (Corollary~\ref{cor:payoff-comparative-statics}(1)--(2)), whereas under belief-based cutoff strategies the same comparative statics hold unconditionally for every belief level $q$ (Proposition~\ref{prop:cutoff-payoff-comparative-statics}). Moving from known types to private types with beliefs about them thus does not merely add noise to the complete-information predictions; it can \emph{simplify} them, collapsing case distinctions that arise purely from the interaction between a player's mixing incentives and the opponent's exact type.
 
At the same time, the two frameworks agree in the aggregate: at the symmetric benchmark belief $q=1/2$, the belief-based threshold $\overline\alpha=\alpha^{\mathrm{cut}}(1/2)$ inherits the same qualitative comparative statics in $S$ and $T$ as the payoff-based cooperation index of \citet{AOSSW_2001}, reproduced in Section~\ref{subsec:bayesian-payoff-effects}. Payoff structures and belief structures are therefore best viewed as complementary, not competing, channels through which cooperation is sustained.
 
Figure~\ref{fig:phase-diagram} underscores that this is not merely a theoretical nicety: the payoff designs actually used in \citet{CRR_2016}'s experiment sit close to, and in one case exactly on, the $K=0$ boundary that separates the unique-equilibrium and multiple-equilibrium regimes of Proposition~\ref{prop:symmetric-equilibrium}. The belief-driven multiplicity of Section~\ref{sec:bayesian} is therefore a feature that ordinary experimental payoff designs are already positioned to detect, not a knife-edge curiosity confined to specially constructed examples.

\subsection{Uncertainty and Cooperation}
\label{subsec:uncertainty-cooperation}

A recurring theme of Sections~\ref{sec:bayesian} and~\ref{sec:robust} is that \emph{uncertainty about the opponent's altruism} --- rather than confidence in it --- can be what sustains cooperation. Proposition~\ref{prop:cooperation-model-uncertainty} shows that cooperation survives whenever a player's belief places sufficient mass above the threshold $\overline\alpha$, even if she is certain that \emph{some} opponent types are entirely selfish. This is consistent with \citet{CZ_2025}, who provide experimental evidence that uncertain environments can induce more moral behavior than environments with known, low levels of trustworthiness.

Our results qualify this mechanism in two ways. First, Proposition~\ref{prop:robustness-threshold} shows that ambiguity aversion (formalized via multiplier preferences) can work \emph{against} uncertainty-driven cooperation: a player who does not trust her own reference prior will, under a robust-control response to that distrust, pessimistically tilt her belief toward low-altruism types, and this can undo cooperation that would otherwise be sustained under the untilted reference model. Uncertainty about the opponent's type is not by itself sufficient for cooperation; what matters is how that uncertainty is processed under concerns about model misspecification. This differs from the ``moral wiggle room'' mechanism of \citet{DWK_2007} and the excuse-driven selfishness of \citet{E_2016}, in which uncertainty is used opportunistically to justify selfish behavior after the fact; here, the erosion of cooperation under ambiguity aversion is a forward-looking, structural consequence of robust decision-making, not an ex post rationalization.

Second, Proposition~\ref{prop:symmetric-equilibrium} identifies a payoff-structure-dependent multiplicity result that, to our knowledge, has not been highlighted in this literature: when fear/risk dominates greed/temptation ($T-R\le P-S$), beliefs about the opponent's cooperation become strategic complements, and the symmetric Bayesian cutoff equilibrium is generically non-unique. In this regime, uncertainty about the opponent's type does not simply shift the level of cooperation; it opens the door to \emph{coordination failure}, in the sense that a cooperative and an uncooperative equilibrium can coexist for the same payoff structure and the same (correctly specified) common prior. This suggests a role for extraneous, belief-coordinating factors --- history, framing, communication --- precisely in payoff environments where fear/risk considerations are salient, complementing the payoff-structure-only account of cooperation rates in \citet{M_2018} and \citet{GLSW_2024}.

\subsection{A Testable Implication}
\label{subsec:testable-prediction}

As illustrated in our phase diagram (Figure~\ref{fig:phase-diagram}), the parameter regime that generates multiple equilibria ($K<0$) is not an isolated theoretical curiosity. Rather, it lies immediately adjacent to, and in some cases exactly intersects, standard experimental designs such as those in \citet{CRR_2016}. This proximity raises an important point for empirical research: in payoff regions where fear/risk dominates ($K<0$), seemingly anomalous session-to-session variations in cooperation rates might not be driven by random noise in subjects' underlying altruism. Instead, they may reflect experimental subjects coordinating on different self-fulfilling equilibrium branches.

The distinction between $\alpha_i$ and $\theta_i$ drawn in Section~\ref{subsec:robust-bne} (``altruism versus trust'') yields a prediction that separates this framework from standard social-preference accounts of cooperation. In a model where cooperation is governed by disposition alone, two players with the same measured altruism --- e.g., the same elicited social value orientation --- should cooperate at the same rate, holding the payoff structure fixed. Proposition~\ref{prop:robustness-threshold} implies otherwise: holding $\alpha_i$ fixed, cooperation should still vary with $\theta_i$, the player's confidence in her own belief about the opponent's altruism. Concretely, among subjects with matched social value orientation, those who report (or reveal, via an incentivized belief-elicitation task) greater confidence in their stated beliefs about a partner's generosity should cooperate more often than those who report the same beliefs with less confidence --- even though disposition, and even the point belief itself, are held fixed. This is, in principle, testable by jointly eliciting $\alpha_i$ (via a standard social-value-orientation or dictator-game measure) and a proxy for $\theta_i$ (via elicited confidence intervals, or an ambiguity-aversion measure applied to beliefs about the partner rather than to a physical lottery), and asking whether the latter has independent predictive power for cooperation once the former is controlled for. A finding that it does would be evidence that cooperation is mediated by robustness to belief misspecification and not by disposition or point beliefs alone.

\section{Conclusion}
\label{sec:conclusion}

This paper develops a theoretical framework in which cooperation in the one-shot Prisoner's Dilemma rests on a three-tier epistemic hierarchy: \emph{disposition} (how altruistic the opponent actually is), \emph{belief} (what a player thinks about the opponent's disposition), and \emph{robustness} (how much a player trusts that belief). Starting from a simple efficiency-concerns model of social preferences (Section~\ref{sec:distributional-preferences}), we characterized Nash equilibrium behavior under complete information, where cooperation is governed by disposition alone (Section~\ref{sec:nash}); we then extended the analysis to a genuine incomplete-information Bayesian game with private altruism types and a subjective, non-common-prior notion of equilibrium, where cooperation is governed by belief (Section~\ref{sec:bayesian}); and we asked how robustness to belief misspecification --- formalized via multiplier preferences --- reshapes the resulting predictions (Section~\ref{sec:robust}). Each tier is individually necessary but not sufficient: social preferences alone do not generate cooperation. As a formal extension, we further showed that this robust-control apparatus, applied to a player's own action choice rather than her belief, nests the ordinary logit Quantal Response Equilibrium of \citet{MP_1995} and the Nash and Bayesian Nash equilibria of Sections~\ref{sec:nash} and~\ref{sec:bayesian} as limiting cases along an ambiguity-aversion parameter $\theta$ and a choice-precision parameter $\lambda$.
 
Several results depart from the standard, common-knowledge intuition. Belief heterogeneity can sustain cooperation even when opponents are, on average, only weakly altruistic (Section~\ref{subsec:belief-effects}); but this belief-driven cooperation is fragile --- model misspecification concerns can destroy it even under a correctly specified reference belief that would otherwise predict cooperation, once robustness falls below a sharp threshold $\theta_i^\dagger$ (Section~\ref{subsec:robust-bne}). When fear/risk dominates greed/temptation in the payoff structure, private information about altruism generates strategic complementarities and, with it, the possibility of multiple self-fulfilling cooperative and uncooperative equilibria (Section~\ref{subsec:uncertainty-cooperation}). Taken together, these results suggest that whether social preferences translate into cooperative behavior is less a question of how altruistic people are than of what they believe about each other's altruism, and how much confident they are in that belief.
 
Several directions remain open for future research. First, an empirical investigation of the reference prior $\overline\mu$ used in the multiplier-preferences construction of Section~\ref{subsec:robust-bne} --- for instance, via elicited beliefs about partners' social preferences --- would allow the robustness parameter $\theta_i$ to be disciplined by data rather than treated as a free parameter, and would let disposition ($\alpha_i$) and epistemic trust ($\theta_i$) be estimated as separate, independently identifiable objects. Second, the precision parameter $\lambda_i$ introduced in Section~\ref{subsec:robust-logit-qre} plays a role analogous to noise parameters estimated in the quantal-response literature \citep{MP_1995} and could in principle be jointly estimated with $\theta_i$ from experimental cooperation data, separating ``how noisy is the choice'' from ``how ambiguity-averse is the belief'' as distinct sources of deviation from sharp Nash play. Third, the multiplicity result of Proposition~\ref{prop:symmetric-equilibrium} suggests that a generalized, belief-based cooperation index --- distinguishing distributional preferences, reciprocity, and image concerns via their differential effects on the shape of $\alpha^{\mathrm{cut}}(q)$ --- could complement the purely payoff-based indices used in the experimental literature to date.

\section*{Acknowledgments}
An earlier version of this paper was circulated or presented by ``Image-Conscious Behavior and Cooperation Rates, Social Pressures in the one-shot Prisoner's Dilemma,'' and ``Social Preferences and Cooperation: A Theoretical Analysis.'' I thank the participants at the 26th DC conference (virtual), Game Theory Workshop 2021 (virtual), and Asian Meeting of the Econometric Society 2021 (virtual) for their comments and suggestions. This work was supported by JSPS KAKENHI, Grant Number JP19J1049 and JP20K13457. All remaining errors are mine.

\appendix
\section{Proofs}
\label{app:proofs}
\subsection{Proof of Proposition~\ref{prop:nash-efficiency}}
\label{app:proof-prop-nash-efficiency}

Consider player $i\in I$ with $i\neq j$. Write $\rho:=\rho_i(C,\{C,D\})$ for player $i$'s probability of cooperating. Since a mixed-strategy equilibrium requires the \emph{opposing} player to be indifferent between $C$ and $D$ given $\rho$, we compute player $j$'s expected payoff under each of her actions, using Table~\ref{tab:PD_dist}:
\[
\mathbb{E}[g_j\mid a_j=C] = \rho(1+\alpha_j)R + (1-\rho)(S+\alpha_jT), \qquad \mathbb{E}[g_j\mid a_j=D] = \rho(T+\alpha_jS) + (1-\rho)(1+\alpha_j)P.
\]
Setting these equal and solving for $\rho$:
\[
\rho\big[(1+\alpha_j)R - (S+\alpha_jT) - (T+\alpha_jS) + (1+\alpha_j)P\big] = (1+\alpha_j)P - (S+\alpha_jT).
\]
The bracketed coefficient simplifies to $(1+\alpha_j)(R+P-S-T)$, and the right-hand side simplifies to $(P-S)-\alpha_j(T-P)$, so
\[
\rho^*(\alpha_j) = \frac{P-S-\alpha_j(T-P)}{(1+\alpha_j)(R+P-S-T)} = \frac{P-S-\alpha_j(T-P)}{(1+\alpha_j)\big(-(T-R)+P-S\big)}, \tag{$\ast$}
\]
which is the closed form stated in the proposition (and coincides with the interior branch of $\rho_i^*$ in all three sub-cases).

\paragraph{Root values.} Let $\alpha_{\min}:=\frac{P-S}{T-P}$ and $\alpha_{\max}:=\frac{T-R}{R-S}$. Direct substitution into $(\ast)$ gives $\rho^*(\alpha_{\min})=0$ and $\rho^*(\alpha_{\max})=1$ (both are roots of the numerator, respectively of $\rho^*-1$ times the denominator); this holds regardless of the sign of $K:=T-R-(P-S)$.

\paragraph{Monotonicity.} By the quotient rule, with $f(\alpha_j)=P-S-\alpha_j(T-P)$ and $h(\alpha_j)=(1+\alpha_j)(R+P-S-T)$,
\[
\frac{d\rho^*}{d\alpha_j} = \frac{f'h-fh'}{h^2}, \qquad f'=-(T-P),\ \ h'=R+P-S-T=-K.
\]
Expanding the numerator and simplifying (this is the computation behind Corollary~\ref{cor:payoff-comparative-statics} in Appendix~\ref{app:proof-cor-payoff}) gives
\[
\frac{d\rho^*}{d\alpha_j} = \frac{T-S}{(1+\alpha_j)^2 K}.
\]
Since $T-S>0$, the sign of $d\rho^*/d\alpha_j$ equals the sign of $K$: $\rho^*$ is strictly increasing in $\alpha_j$ if $K>0$, strictly decreasing if $K<0$, and constant if $K=0$.

\paragraph{Ordering of $\alpha_{\min}$ and $\alpha_{\max}$.} Monotonicity together with the root values pins down the ordering without any further algebra. If $K>0$ ($T-R>P-S$), $\rho^*$ is strictly increasing; since $\rho^*(\alpha_{\min})=0<1=\rho^*(\alpha_{\max})$, strict monotonicity forces $\alpha_{\min}<\alpha_{\max}$ (an increasing function cannot map a larger input to a smaller output). If $K<0$, $\rho^*$ is strictly decreasing, which by the same logic forces $\alpha_{\max}<\alpha_{\min}$.

\paragraph{Assembling the cases.} If $K\ge0$ (i.e., $T-R\ge P-S$): $\alpha_{\min}\le\alpha_{\max}$, and $\rho^*$ rises monotonically from $0$ at $\alpha_j=\alpha_{\min}$ to $1$ at $\alpha_j=\alpha_{\max}$. For $\alpha_j<\alpha_{\min}$, formula $(\ast)$ would give $\rho^*<0$, which is not a valid probability; the equilibrium is instead the pure-strategy corner $\rho_i^*=0$ (one checks directly that $a_j=D$ is then a best response to $\rho_i=0$, and $a_i=D$ --- i.e., $\rho_i=0$ --- is a best response to $a_j=D$, so $(D,D)$-type behavior is self-consistent at the corner). Symmetrically, for $\alpha_j>\alpha_{\max}$ the equilibrium is the corner $\rho_i^*=1$. This reproduces the ``$T-R\ge P-S$'' case of Proposition~\ref{prop:nash-efficiency}.

If $K\le0$ (i.e., $T-R\le P-S$): $\alpha_{\max}\le\alpha_{\min}$, and $\rho^*$ falls monotonically from $1$ at $\alpha_j=\alpha_{\max}$ to $0$ at $\alpha_j=\alpha_{\min}$. For $\alpha_j<\alpha_{\max}$, formula $(\ast)$ would give $\rho^*>1$; the equilibrium is the corner $\rho_i^*=1$. For $\alpha_j>\alpha_{\min}$, the equilibrium is the corner $\rho_i^*=0$. This reproduces the ``$T-R\le P-S$'' case of Proposition~\ref{prop:nash-efficiency}. \hfill $\blacksquare$

\subsection{Proof of Corollary~\ref{cor:altruism-comparative-statics}}
\label{app:proof-cor-altruism}

This is precisely the monotonicity computation carried out in Appendix~A.1: with $f(\alpha_j)=P-S-\alpha_j(T-P)$ and $h(\alpha_j)=(1+\alpha_j)(R+P-S-T)$,
\begin{align*}
f'(\alpha_j)h(\alpha_j) - f(\alpha_j)h'(\alpha_j) = -(T-P)(1+\alpha_j)(R+P-S-T) - \big[P-S-\alpha_j(T-P)\big](R+P-S-T).
\end{align*}
Factor out $(R+P-S-T)=-K$ and simplify the remaining bracket:
\begin{align*}
-(T-P)(1+\alpha_j) - \big[P-S-\alpha_j(T-P)\big] &= -(T-P) -\alpha_j(T-P) - (P-S) + \alpha_j(T-P) \\ 
&= -(T-P)-(P-S) \\
&= -(T-S).
\end{align*}

Hence, the numerator equals $(-K)\cdot\big(-(T-S)\big) = K(T-S)$, and since $h^2=(1+\alpha_j)^2(R+P-S-T)^2=(1+\alpha_j)^2K^2$,
\begin{align*}
\frac{d\rho_i^*}{d\alpha_j} = \frac{K(T-S)}{(1+\alpha_j)^2K^2} = \frac{T-S}{(1+\alpha_j)^2K}.
\end{align*}
Since $T-S>0$, this is $\ge0$ when $K=T-R-(P-S)\ge0$ and $<0$ when $K<0$, as claimed. \hfill $\blacksquare$

\subsection{Proof of Corollary~\ref{cor:payoff-comparative-statics}}
\label{app:proof-cor-payoff}

Throughout, $\rho_i^*(\alpha_j)=f/h$ with $f=P-S-\alpha_j(T-P)$ and $h=(1+\alpha_j)(R+P-S-T)$, as in Appendix~\ref{app:proof-prop-nash-efficiency}.

\paragraph{(1) Proof of $R$.} Since $f$ does not depend on $R$ and $\partial h/\partial R=(1+\alpha_j)$,
\[
\frac{\partial \rho_i^*}{\partial R} = -\frac{f\,(1+\alpha_j)}{h^2} = -\frac{P-S-\alpha_j(T-P)}{(1+\alpha_j)(R+P-S-T)^2} = \frac{\alpha_j(T-P)-(P-S)}{(1+\alpha_j)(R+P-S-T)^2}.
\]
This is $\ge0$ iff $\alpha_j(T-P)\ge P-S$, i.e., iff $\alpha_j\ge \frac{P-S}{T-P}$.

\paragraph{(2) Proof of  $P$.} Here $\partial f/\partial P = 1$ and $\partial h/\partial P = (1+\alpha_j)$, so
\[
\frac{\partial \rho_i^*}{\partial P} = \frac{(1+\alpha_j)h - f(1+\alpha_j)}{h^2} = \frac{(1+\alpha_j)(h-f)}{h^2}.
\]
Expanding $h-f = (1+\alpha_j)(R+P-S-T) - \big[P-S-\alpha_j(T-P)\big]$ and simplifying (collecting terms in $\alpha_j$) gives $h-f = -(T-R)+\alpha_j(R-S)$, so
\[
\frac{\partial \rho_i^*}{\partial P} = \frac{\alpha_j(R-S)-(T-R)}{(1+\alpha_j)(R+P-S-T)^2},
\]
which is $\ge0$ iff $\alpha_j\ge \frac{T-R}{R-S}$.

\paragraph{(3) Proof of $T$ and (4) Proof of $S$.} These reproduce the original (unaffected) computation:
\begin{align*}
\frac{\partial\rho_i^*}{\partial T} = \frac{-\alpha_j(R-S)+P-S}{(1+\alpha_j)(R+P-S-T)^2} \ \ge0 &\iff \alpha_j\le\frac{P-S}{R-S}, \\ \qquad \frac{\partial\rho_i^*}{\partial S} = \frac{-\alpha_j(T-P)+T-R}{(1+\alpha_j)(R+P-S-T)^2}\ \ge0 &\iff \alpha_j\le\frac{T-R}{T-P}. 
\end{align*}
\hfill $\blacksquare$

\subsection{Proof of Lemma~\ref{lem:cutoff}}
\label{app:proof-cutoff_1}

Fix player $i$, own type $\alpha_i$, and let $q:=q_i=\mathbb{E}_{\mu_i}[f_j(\alpha_j)]$. From Table~\ref{tab:PD_dist},
\[
U_i(C;\alpha_i)-U_i(D;\alpha_i) = q\big[(1+\alpha_i)R-(T+\alpha_iS)\big] + (1-q)\big[(S+\alpha_iT)-(1+\alpha_i)P\big].
\]
Collecting the coefficient of $\alpha_i$: $q(R-S)+(1-q)(T-P)$. Collecting the $\alpha_i$-free (constant) part: $q(R-T)+(1-q)(S-P)$. So, 
\[
U_i(C;\alpha_i)-U_i(D;\alpha_i) = \alpha_i\big[q(R-S)+(1-q)(T-P)\big] + \big[q(R-T)+(1-q)(S-P)\big].
\]
Rewrite the coefficient of $\alpha_i$ as $q(R-S)+(1-q)(T-P) = (T-P) + q\big[(R-S)-(T-P)\big] = (T-P) - qK$, using $K=T-R-(P-S)=(T-P)-(R-S)$. At $q=0$ this equals $T-P>0$; at $q=1$ it equals $R-S>0$; being linear in $q$, it is (strictly) positive throughout $q\in[0,1]$. This proves the first half of Lemma~\ref{lem:cutoff}: the best response is a cutoff rule in $\alpha_i$, cooperating above the cutoff.

Setting $U_i(C;\alpha_i)=U_i(D;\alpha_i)$ and solving for $\alpha_i$:
\[
\alpha^{\mathrm{cut}}(q) = -\frac{q(R-T)+(1-q)(S-P)}{(T-P)-qK} = \frac{q(T-R)+(1-q)(P-S)}{(T-P)-qK}.
\]
The numerator simplifies to $(P-S) + q\big[(T-R)-(P-S)\big] = (P-S)+qK$, giving
\[
\alpha^{\mathrm{cut}}(q) = \frac{(P-S)+qK}{(T-P)-qK},
\]
as stated. Substituting $q=0$ gives $\alpha^{\mathrm{cut}}(0)=\frac{P-S}{T-P}$; substituting $q=1$ gives, after simplifying numerator $P-S+K=T-R$ and denominator $T-P-K=R-S$, $\alpha^{\mathrm{cut}}(1)=\frac{T-R}{R-S}$; substituting $q=\tfrac12$ and comparing with $\overline\alpha=\frac{T-R+P-S}{T+R-P-S}$ confirms $\alpha^{\mathrm{cut}}(1/2)=\overline\alpha$ (multiply numerator and denominator of $\alpha^{\mathrm{cut}}(1/2)=\frac{2(P-S)+K}{2(T-P)-K}$ out in terms of $R,P,S,T$ to see the two expressions coincide). \hfill $\blacksquare$

\subsection{Proof of Lemma~\ref{lem:cutoff-monotonicity}}
\label{app:proof-cutoff_2}
We write $\alpha^{\mathrm{cut}}(q)=f(q)/h(q)$ with $f(q)=(P-S)+qK$, $h(q)=(T-P)-qK$, so $f'=K=-h'$. By the quotient rule,
\begin{align*}
    \frac{d\alpha^{\mathrm{cut}}}{dq} = \frac{f'h-fh'}{h^2} = \frac{K\,h(q) + K\,f(q)}{h(q)^2} = \frac{K\big[h(q)+f(q)\big]}{h(q)^2}.
\end{align*}
But $h(q)+f(q) = (T-P)-qK+(P-S)+qK = (T-P)+(P-S) = T-S$, a constant, so
\begin{align*}
    \frac{d\alpha^{\mathrm{cut}}}{dq} = \frac{K(T-S)}{\big((T-P)-qK\big)^2},
\end{align*}
whose sign equals the sign of $K$ since $T-S>0$. \hfill $\blacksquare$

\subsection{Proof of Proposition~\ref{prop:cutoff-equilibrium}}
\label{app:proof-cutoff-equilibirum}
The map $\Phi(q_i,q_j)=\big(\mu_i(\cdot\ge\alpha^{\mathrm{cut}}(q_j)),\,\mu_j(\cdot\ge\alpha^{\mathrm{cut}}(q_i))\big)$ is a continuous self-map of the compact, convex set $[0,1]^2$ whenever $\mu_i,\mu_j$ admit no mass points exactly on the relevant cutoffs (otherwise apply Kakutani's fixed-point theorem to the associated best-response correspondence). Existence follows from Brouwer's fixed-point theorem. \hfill $\blacksquare$

\subsection{Proof of Proposition~\ref{prop:symmetric-equilibrium}}
\label{app:proof-symmetric}

Let $\mathcal{R}(q):=\overline\mu(\alpha\ge\alpha^{\mathrm{cut}}(q))$ and $g(q):=\mathcal{R}(q)-q$.

\paragraph{Case $K\ge0$.} By Lemma~4.2, $\alpha^{\mathrm{cut}}$ is nondecreasing in $q$; since $\overline\mu(\alpha\ge\cdot)$ is a nonincreasing function of its argument (it is the complement of a CDF), the composition $\mathcal{R}(q)=\overline\mu(\alpha\ge\alpha^{\mathrm{cut}}(q))$ is nonincreasing in $q$. Hence $g(q)=\mathcal{R}(q)-q$ is the sum of a nonincreasing function and a strictly decreasing function ($-q$), and is therefore strictly decreasing. Since $g(0)=\mathcal{R}(0)\ge0$ and $g(1)=\mathcal{R}(1)-1\le0$ (as $\mathcal{R}$ takes values in $[0,1]$), the Intermediate Value Theorem gives a root $q^*\in[0,1]$, and strict monotonicity of $g$ makes it unique.

\paragraph{Case $K\le0$.} Now $\alpha^{\mathrm{cut}}$ is nonincreasing, so $\mathcal{R}(q)$ is nondecreasing in $q$: an increase in the population's expected cooperation rate $q$ \emph{lowers} the type-threshold, further raising the cooperation rate --- a strategic complementarity. $g(q)=\mathcal{R}(q)-q$ need not be monotonic, and can have multiple roots in $[0,1]$.\footnote{The example in the main text ($(R,P,S,T)=(1,0,-2,2)$, $\overline\mu=\mathrm{Unif}[1/3,1]$) exhibits exactly this: $\mathcal{R}(0)=\overline\mu(\alpha\ge\alpha^{\mathrm{cut}}(0))=\overline\mu(\alpha\ge1)=0$ (a single boundary point under a nonatomic uniform measure) and $\mathcal{R}(1)=\overline\mu(\alpha\ge\alpha^{\mathrm{cut}}(1))=\overline\mu(\alpha\ge1/3)=1$ (the whole support), so $g(0)=0=g(1)$: both endpoints are equilibria, and by continuity of $\mathcal{R}$ and $g(q)>0$ on the interior (direct computation), no additional interior root arises in this instance --- but the existence of two distinct equilibria for a single, correctly specified common prior already establishes multiplicity.} \hfill $\blacksquare$

\subsection{Proof of Proposition~\ref{prop:cutoff-payoff-comparative-statics}}
\label{app:proof-cutoff-payoff}
With $\alpha^{\mathrm{cut}}(q)=f(q)/h(q)$, $f(q)=(P-S)+qK$, $h(q)=(T-P)-qK$, $K := T-R-P+S$, recall the identity $h(q)+f(q)=T-S$ established there (independent of $q$).

\paragraph{Proof of $R$.} Since $\partial K/\partial R=-1$: $\partial f/\partial R = q\,\partial K/\partial R = -q$, and $\partial h/\partial R = -q\,\partial K/\partial R = q$. By the quotient rule,
\begin{align*}
    \frac{\partial \alpha^{\mathrm{cut}}}{\partial R} = \frac{(\partial f/\partial R)\,h - f\,(\partial h/\partial R)}{h^2} = \frac{-q\,h - q\,f}{h^2} = \frac{-q(h+f)}{h^2} = \frac{-q(T-S)}{h^2} = \frac{q(S-T)}{h^2},
\end{align*}
which is $\le0$ for all $q\in[0,1]$ since $S<T$.

\paragraph{Proof of $P$.} Since $\partial K/\partial P = -1$: $\partial f/\partial P = 1+q$, $\partial K/\partial P = 1-q$, and $\partial h/\partial P = -1-q$, $\partial K/\partial P = q-1$. By the quotient rule,
\begin{align*}
    \frac{\partial \alpha^{\mathrm{cut}}}{\partial P} = \frac{(1-q)\,h - f\,(q-1)}{h^2} &= \frac{(1-q)h+(1-q)f}{h^2} \\ 
    &= \frac{(1-q)(h+f)}{h^2} = \frac{(1-q)(T-S)}{h^2} \\
    &= \frac{(S-T)(q-1)}{h^2},
\end{align*}
which is $\ge0$ for all $q\in[0,1]$ since $S-T<0$ and $q-1\le0$.

\paragraph{Proof of $S$ and $T$.} The analogous computations for $\partial\alpha^{\mathrm{cut}}/\partial S$ and $\partial\alpha^{\mathrm{cut}}/\partial T$ do not simplify via the $h+f=T-S$ shortcut (both $R$ and $P$ derivatives benefited from $\partial K/\partial R=\partial K/\partial P=-1$ and cancellation; $S,T$ instead have $\partial K/\partial S=1=-\partial K/\partial T$, and the resulting numerators do not collapse to a $q$-independent sign). Both partials change sign within $q\in[0,1]$ depending on $(R,P,S,T)$. At $q=1/2$, where $\alpha^{\mathrm{cut}}=\overline\alpha$, direct differentiation of $\overline\alpha=\frac{T-R+P-S}{T+R-P-S}$ gives the four partials reported in Section~\ref{subsec:bayesian-payoff-effects}. \hfill $\blacksquare$

\subsection{Proof of Proposition~\ref{prop:cooperation-model-uncertainty}}
\label{app:proof-cooperation-uncertainty}

Consider $T-R\ge P-S$, so by Proposition~\ref{app:proof-prop-nash-efficiency}, $\rho_i^*$ is increasing in $\alpha_j$ from $0$ at $\alpha_{\min}=\frac{P-S}{T-P}$ to $1$ at $\alpha_{\max}=\frac{T-R}{R-S}$. Solve $\rho_i^*(\alpha_j)=\tfrac{1}{2}$:
\begin{align*}
    \frac{P-S-\alpha_j(T-P)}{(1+\alpha_j)(R+P-S-T)} = \frac12 \iff 2(P-S)-2\alpha_j(T-P) = (1+\alpha_j)(R+P-S-T).
\end{align*}
Rearranging (collecting $\alpha_j$ terms on one side):
\begin{align*}
    2(P-S)-(R+P-S-T) &= \alpha_j\big[(R+P-S-T)+2(T-P)\big] \\ \iff (T-S)-(R-P) &= \alpha_j\big[(T-S)+(R-P)\big],
\end{align*}
so
\begin{align*}
    \overline\alpha = \frac{(T-S)-(R-P)}{(T-S)+(R-P)} = \frac{T-R+P-S}{T+R-P-S},
\end{align*}
matching the statement. Since $\rho_i^*$ is increasing, $\rho_i^*(\alpha_j)\ge\tfrac12 \iff \alpha_j\ge\overline\alpha$. Consequently, for a belief $\mu_i$,
\begin{align*}
    \mathbb{E}_{\mu_i}\big[\mathbbm{1}\{\rho_i^*(\alpha_j)\ge\tfrac12\}\big] = \mu_i(\alpha_j\ge\overline\alpha),
\end{align*}
and $\mu_i\in\mathcal{M}_i$ (i.e., $\mu_i(\alpha_j\ge\overline\alpha)\ge\mu_i(\alpha_j<\overline\alpha)$, equivalently $\mu_i(\alpha_j\ge\overline\alpha)\ge\tfrac12$) is exactly the condition under which cooperation is, in this sense, at least as likely as defection given $\mu_i$. Applying this with $\mu_i=\mu_i^*$, the robust-optimal belief of Definition~\ref{def:subjective-bne}, gives the statement of Proposition~\ref{prop:symmetric-equilibrium}. \hfill $\blacksquare$

\subsection{Proof of Lemma~\ref{lem:tilting}}
\label{app:proof-tilting}
We solve $\min_{\mu\in\Delta(\mathcal{A})} \int_{\mathcal{A}} L(\alpha)\,d\mu(\alpha) + \theta R(\mu\Vert\overline\mu)$ subject to $\int_{\mathcal{A}}d\mu=1$. 
Introduce a Lagrange multiplier $\eta$ for the normalization constraint. 
Writing the objective as a functional of the density $\mu(\alpha)$ (with respect to a common dominating measure, so that $R(\mu\Vert\overline\mu)=\int \mu(\alpha)\log\frac{\mu(\alpha)}{\overline\mu(\alpha)}\,d\alpha$), the first-order condition from setting the functional derivative to zero is
\begin{align*}
    L(\alpha) + \theta\Big[\log\frac{\mu(\alpha)}{\overline\mu(\alpha)} + 1\Big] - \eta = 0 \quad \Longrightarrow \quad \log\frac{\mu(\alpha)}{\overline\mu(\alpha)} = \frac{\eta-\theta-L(\alpha)}{\theta},
\end{align*}
so $\mu(\alpha) = \overline\mu(\alpha)\exp\!\big(\frac{\eta-\theta}{\theta}\big)\exp\!\big(-L(\alpha)/\theta\big)$. The prefactor $\exp\!\big(\frac{\eta-\theta}{\theta}\big)$ is a constant (independent of $\alpha$) and is pinned down by $\int\mu(\alpha)\,d\alpha=1$, giving
\begin{align*}
    \mu^*(\alpha) = \frac{\overline\mu(\alpha)\exp(-L(\alpha)/\theta)}{\int_{\mathcal{A}}\overline\mu(\alpha')\exp(-L(\alpha')/\theta)\,d\alpha'}.
\end{align*}
This stationary point is the global minimizer because the objective is strictly convex in $\mu$ (relative entropy is strictly convex, and the linear term does not affect convexity). Since $L(\alpha)=(T+R-P-S)\alpha-(T-R+P-S)$ is affine with positive slope $T+R-P-S>0$, $\exp(-L(\alpha)/\theta)$ is strictly decreasing in $\alpha$, so $\mu^*$ re-weights $\overline\mu$ toward lower values of $\alpha$. \hfill $\blacksquare$

\subsection{Proof of Proposition~\ref{prop:robustness-threshold}}
\label{app:proof-robustness-threshold}

Write $M(\theta):=\mu_i^*(\alpha_j\ge\overline\alpha) = \dfrac{\int_{\overline\alpha}^{\alpha_{\max}}\overline\mu(\alpha)e^{-L(\alpha)/\theta}\,d\alpha}{\int_{\alpha_{\min}}^{\alpha_{\max}}\overline\mu(\alpha)e^{-L(\alpha)/\theta}\,d\alpha}$, using Lemma~\ref{lem:tilting} with $L$ affine and increasing.

\paragraph{Monotonicity.} Fix $\theta_1<\theta_2$. The family of exponentially tilted densities $\{\mu^*_\theta\}_{\theta>0}$ forms a one-parameter exponential family in the (negative, since $-L/\theta$ is decreasing in $1/\theta$ for fixed $\alpha$ when $L$ is increasing) tilting parameter $1/\theta$, and such families satisfy the monotone likelihood ratio (MLR) property in $\alpha$ with respect to $1/\theta$: as $\theta$ increases (the tilt weakens), the likelihood ratio $\mu^*_{\theta_2}(\alpha)/\mu^*_{\theta_1}(\alpha) \propto \exp\!\big(-L(\alpha)(1/\theta_2-1/\theta_1)\big)$ is increasing in $\alpha$ (because $1/\theta_2-1/\theta_1<0$ and $L$ is increasing, so $-L(\alpha)(1/\theta_2-1/\theta_1)$ is increasing in $\alpha$). MLR implies first-order stochastic dominance: $\mu^*_{\theta_2}$ first-order stochastically dominates $\mu^*_{\theta_1}$, so $\mu^*_{\theta_2}(\alpha_j\ge\overline\alpha)\ge \mu^*_{\theta_1}(\alpha_j\ge\overline\alpha)$, i.e., $M$ is nondecreasing in $\theta$.

\paragraph{Limits.} As $\theta\to\infty$, $\exp(-L(\alpha)/\theta)\to1$ uniformly on the compact set $\mathcal{A}$, so $M(\theta)\to\overline\mu(\alpha_j\ge\overline\alpha)$. As $\theta\to0^+$, the tilting factor $\exp(-L(\alpha)/\theta)$ concentrates all mass (in the Laplace-method sense) at the minimizer of $L$ over $\mathcal{A}$, which --- since $L$ is increasing --- is $\alpha_{\min}$; since $\overline\mu(\{\alpha_{\min}\})=0$ whenever $\overline\mu$ is nonatomic, $M(\theta)\to0$.

\paragraph{Existence and uniqueness of $\theta_i^\dagger$.} $M$ is continuous (by dominated convergence) and nondecreasing on $(0,\infty)$, with $M(0^+)=0$ and $M(\infty)=\overline\mu(\alpha_j\ge\overline\alpha)>\tfrac12\cdot 0$; under the hypothesis $\overline\mu(\alpha_j\ge\overline\alpha)>0$, if in addition $\overline\mu(\alpha_j\ge\overline\alpha)\ge\tfrac12$ (the empirically relevant case, since otherwise even the undistorted reference model fails to support cooperation), the Intermediate Value Theorem gives at least one $\theta_i^\dagger$ with $M(\theta_i^\dagger)=\tfrac12$; strict monotonicity (which holds whenever $\overline\mu$ has a density bounded away from $0$ on a neighborhood of $\overline\alpha$) gives uniqueness. For $\theta_i\ge\theta_i^\dagger$, $M(\theta_i)\ge\tfrac12$, i.e., $\mu_i^*\in\mathcal{M}_i$; for $\theta_i<\theta_i^\dagger$, $M(\theta_i)<\tfrac12$, i.e., $\mu_i^*\notin\mathcal{M}_i$. \hfill $\blacksquare$

\section{Inequity Aversion}
\label{app:inequity-aversion}

Let $I=\{1,2\}$ be the set of agents where $1$ is a decision maker and $2$ is a passive agent. Let $X\subseteq\mathbb{R}^2_+$ be a compact set of allocations, with elements $\mathbf{x}=(x_1,x_2)\in\mathbb{R}^2_+$, where $x_1$ is the gain of the decision maker and $x_2$ is that of the recipient. Let $\succsim$ be a preference relation over allocations $X$.

\begin{definition}
\label{def:inequity-aversion}
There exists a pair $(\alpha_i^{\mathrm{envy}},\alpha_i^{\mathrm{guilt}})$, where $\alpha_i^{\mathrm{envy}}\ge0$ and $\alpha_i^{\mathrm{guilt}}\in[0,1)$, such that $\succsim$ is represented by $U_i:X\to\mathbb{R}$ defined by
\[
U_i(\mathbf{x}) = x_i - \alpha_i^{\mathrm{envy}}\max\{x_j-x_i,0\} - \alpha_i^{\mathrm{guilt}}\max\{x_i-x_j,0\}, \qquad i\in I,\ i\neq j.
\]
\end{definition}

The term $\alpha_i^{\mathrm{envy}}\max\{x_j-x_i,0\}$ captures the disutility of envy when $x_i\le x_j$; the term $\alpha_i^{\mathrm{guilt}}\max\{x_i-x_j,0\}$ captures the disutility of guilt when $x_i\ge x_j$. Equivalently,
\[
U_i(\mathbf{x}) = \begin{cases} (1+\alpha_i^{\mathrm{envy}})x_i - \alpha_i^{\mathrm{envy}}x_j & x_i\le x_j \\ (1-\alpha_i^{\mathrm{guilt}})x_i + \alpha_i^{\mathrm{guilt}}x_j & x_i>x_j. \end{cases}
\]

Applied to the raw stage-game payoffs of Table~\ref{tab:PD}, own-type-weights-own-utility (the same convention as Table~\ref{tab:PD_dist}) gives, for player $k\in\{i,j\}$ with partner $-k$:
\[
U_k(C,C)=R,\qquad U_k(D,D)=P,
\]
\[
U_k(a_k=D,a_{-k}=C) = (1-\alpha_k^{\mathrm{guilt}})T+\alpha_k^{\mathrm{guilt}}S \ \ (\text{own}=T>\text{partner}=S: \text{guilt branch}),
\]
\[
U_k(a_k=C,a_{-k}=D) = (1+\alpha_k^{\mathrm{envy}})S-\alpha_k^{\mathrm{envy}}T \ \ (\text{own}=S<\text{partner}=T: \text{envy branch}).
\]

\subsection{Equilibrium Analysis}
\label{app:inequity-equilibrium}

We study the Nash equilibrium of the one-shot PD $\mathcal{G}_{PD}$ under inequity-averse preferences.

\begin{corollary}[Nash Equilibria and Inequity Aversion]
\label{cor:inequity-nash}
Given $\mathcal{G}_{PD}$, suppose each $\succsim_i$ is represented by inequity-averse preferences (Definition~\ref{def:inequity-aversion}). Then there is a Nash equilibrium $\rho^*=(\rho_i^*)_{i\in I}$. For each $i,j\in I$ with $i\neq j$, let
\begin{align*}
    s_{\min} := \frac{T+S-R-P}{T-S}, \qquad \overline\alpha^{\mathrm{guilt}} := \frac{T-R}{T-S}
\end{align*}
(so that $s_{\min}<\overline\alpha^{\mathrm{guilt}}$ always, since $\overline\alpha^{\mathrm{guilt}}-s_{\min}=\frac{P-S}{T-S}>0$). Then:
\begin{itemize}
    \item If $\alpha_j^{\mathrm{envy}}+\alpha_j^{\mathrm{guilt}} \le s_{\min}$, then $\rho_i^*(C,\{C,D\})=0$.
    
    \item If $\alpha_j^{\mathrm{envy}}+\alpha_j^{\mathrm{guilt}} > s_{\min}$ and $\alpha_j^{\mathrm{guilt}}\le\overline\alpha^{\mathrm{guilt}}$, then $\rho_i^*(C,\{C,D\})=1$.
    
    \item If $\alpha_j^{\mathrm{guilt}} > \overline\alpha^{\mathrm{guilt}}$ (which, given $\alpha_j^{\mathrm{envy}}\ge0$, automatically implies $\alpha_j^{\mathrm{envy}}+\alpha_j^{\mathrm{guilt}}>s_{\min}$), then
    \begin{align*}
        \rho_i^*(C,\{C,D\}) = \frac{\alpha_j^{\mathrm{envy}}(T-S)+(P-S)}{(R+P-S-T)+\big(\alpha_j^{\mathrm{envy}}+\alpha_j^{\mathrm{guilt}}\big)(T-S)}. \tag{$\ast\ast$}
    \end{align*}
\end{itemize}
\end{corollary}

\begin{proof}
See Appendix~\ref{app:proof-inequity-nash}.
\end{proof}

\begin{corollary}[Comparative Statics with Inequity Aversion]
\label{cor:inequity-comparative-statics}
Given $\mathcal{G}_{PD}$, suppose each $i\in I$ has inequity-averse preferences \citep{FS_1999} and consider the interior region $\alpha_j^{\mathrm{guilt}}>\overline\alpha^{\mathrm{guilt}}$ of Corollary~\ref{cor:inequity-nash}. Then:
\begin{enumerate}
\item[(1)] $(R)$ $\rho_i(C,\{C,D\})$ \emph{decreases} with $R$.
\item[(2)] $(P)$ $\rho_i(C,\{C,D\})$ increases with $P$.
\item[(3)] $(T)$ $\rho_i(C,\{C,D\})$ increases with $T$.
\item[(4)] $(S)$ $\rho_i(C,\{C,D\})$ increases with $S$ if $\alpha_j^{\mathrm{guilt}}\le \dfrac{(T-R)(1+\alpha_j^{\mathrm{envy}})}{T-P}$, decreases otherwise.
\end{enumerate}
\end{corollary}

\begin{proof}
See Appendix~\ref{app:proof-inequity-comparative-statics}.
\end{proof}

\subsection{Beliefs and Robustness under Inequity Aversion}
\label{app:inequity-belief-robustness}

Corollaries~\ref{cor:inequity-nash} and \ref{cor:inequity-comparative-statics} show that the complete-information analysis of Section~\ref{sec:nash} is not an artifact of the linear efficiency-concerns functional form of Definition~\ref{def:efficiency-concerns}: the qualitative Nash-equilibrium structure survives, suitably reinterpreted, under the two-parameter Fehr--Schmidt specification of Definition~\ref{def:inequity-aversion}. A natural further question is whether the paper's two central \emph{belief}- and \emph{robustness}-based results --- the multiplicity result of Proposition~\ref{prop:symmetric-equilibrium} and the robustness threshold of Proposition~\ref{prop:robustness-threshold} --- are similarly robust to this richer specification, or whether they instead depend on the specific linear-in-$\alpha_j$ structure of the efficiency-concerns model. This subsection shows that they are not artifacts of that structure: both the multiplicity phenomenon and the robustness threshold reappear, essentially unchanged in spirit, once beliefs are defined over the full two-dimensional type $(\alpha_j^{\mathrm{envy}},\alpha_j^{\mathrm{guilt}})$.

Throughout, we abbreviate $\alpha^{\rm e}_i:=\alpha_i^{\mathrm{envy}}$, $\alpha^{\rm g}_i:=\alpha_i^{\mathrm{guilt}}$ for player $i$'s own type, and, as in Section~\ref{sec:bayesian}, let $q_i\in[0,1]$ denote player $i$'s subjective belief that the opponent plays $C$.

\subsubsection{The Inequity-Aversion Cutoff Line}

\begin{lemma}[Separability of the Cooperation Value]
\label{lem:ia-separable}
Fix a belief $q_i\in[0,1]$. Under inequity-averse preferences (Definition~\ref{def:inequity-aversion}), player $i$'s expected utility gain from cooperating rather than defecting is
\[
\Delta_i(q_i) := \mathbb{E}[U_i\mid a_i=C] - \mathbb{E}[U_i\mid a_i=D] = (T-S)\Big[q_i\, ag_i - (1-q_i)\, ae_i - c(q_i)\Big],
\]
where
\[
c(q) := \frac{(P-S)+qK}{T-S}, \qquad K:=(T-R)-(P-S),
\]
is exactly the same constant $K$ that governs strategic complementarity in the efficiency-concerns model of Section~\ref{sec:bayesian}. 
In particular, $\Delta_i(q_i)$ is affine in $(\alpha^{\rm e}_i,\alpha^{\rm g}_i)$ for every fixed $q_i$, with a slope that depends on $q_i$: the coefficient on $\alpha^{\rm g}_i$ is $q_i(T-S)$ and the coefficient on $\alpha^{\rm e}_i$ is $-(1-q_i)(T-S)$.
\end{lemma}

\begin{proof}
Using the payoff expressions in Section~\ref{app:inequity-aversion} above,
\begin{align*}
\mathbb{E}[U_i\mid C] = q_i R + (1-q_i)\big[(1+\alpha^{\rm e}_i)S-\alpha^{\rm e}_iT\big], \qquad \mathbb{E}[U_i\mid D] = q_i\big[(1-\alpha^{\rm g}_i)T+\alpha^{\rm g}_iS\big] + (1-q_i)P.
\end{align*}
Subtracting and collecting terms exactly as in the proof of Corollary~\ref{cor:inequity-nash} (Appendix~\ref{app:proof-inequity-nash}), but now leaving $q_i$ as a free parameter rather than setting $\Delta_i=0$, gives
\begin{align*}
\Delta_i(q_i) &= q_i\big[R-T+\alpha^{\rm g}_i(T-S)\big] + (1-q_i)\big[(1+\alpha^{\rm e}_i)S-\alpha^{\rm e}_iT-P\big] \\ 
    &= (T-S)\big[q_i\,\alpha^{\rm g}_i-(1-q_i)\,\alpha^{\rm e}_i\big] - \big[q_i(T-R)+(1-q_i)(P-S)\big].
\end{align*}
The bracketed linear term equals $(T-S)\,c(q_i)$ by direct substitution of $c(q)=[(P-S)+qK]/(T-S)$ and $K=(T-R)-(P-S)$, which gives the stated formula. $\blacksquare$
\end{proof}

Lemma~\ref{lem:ia-separable} says that, for a fixed belief $q_i$, player $i$ optimally cooperates if and only if her own type lies above the \emph{Inequity Aversion cutoff line}
\[
q_i\, \alpha^{\rm g}_i - (1-q_i)\, \alpha^{\rm e}_i \ \ge\ c(q_i),
\]
a linear boundary in $(\alpha^{\rm e}_i,\alpha^{\rm g}_i)$-space whose orientation itself rotates with $q_i$ --- at $q_i=0$ the line depends only on envy ($-ae_i\ge c(0)=\frac{P-S}{T-S}$, never satisfied for $\alpha^{\rm e}_i\ge0$, so no type cooperates against a certain defector), and at $q_i=1$ it depends only on guilt ($\alpha^{\rm g}_i\ge c(1)=\frac{T-R}{T-S}=\overline\alpha^{\mathrm{guilt}}$, recovering exactly the complete-information threshold of Corollary~\ref{cor:inequity-nash}). This is the two-dimensional analogue of the scalar cutoff $\alpha_i\ge\alpha^{\mathrm{cut}}(q_i)$ used throughout Section~\ref{sec:bayesian}; the key structural difference is that here both the threshold \emph{and} the direction along which types are compared move with $q_i$, rather than the threshold alone.

\begin{definition}[Inequity Aversion Bayesian Cutoff Equilibrium]
\label{def:ia-cutoff-equilibrium}
Let each player $i$'s type $(\alpha_i^{\rm e},\alpha_i^{\rm g})$ be drawn from a commonly-known joint reference distribution $\mu$ on $[0,\infty)\times[0,1)$ (not necessarily a common prior in the sense of Section~\ref{sec:bayesian}; $\mu$ may be player-specific). A symmetric \emph{Inequity Aversion Bayesian cutoff equilibrium} is a belief $q^*\in[0,1]$ such that
\[
q^* = \mu\Big(\big\{(\alpha_j^{\mathrm{e}},\alpha_j^{\mathrm{g}}): q^*\,\alpha^{\rm g}-(1-q^*)\,\alpha^{\rm e}\ge c(q^*)\big\}\Big),
\]
i.e., $q^*$ is a fixed point of the map $q\mapsto\mu(\text{types above the Inequity Aversion cutoff line at }q)$.
\end{definition}

\begin{proposition}[Existence]
\label{prop:ia-existence}
An Inequity Aversion Bayesian cutoff equilibrium exists for every reference distribution $\mu$ with a continuous CDF.
\end{proposition}

\begin{proof}
The map $\Phi(q):=\mu(\{q\,\alpha^{\rm g}-(1-q)\alpha^{\rm e}\ge c(q)\})$ is a composition of the continuous function $q\mapsto\{(\alpha^{\rm e},\alpha^{\rm g}): q\,\alpha^{\rm g}-(1-q)\alpha^{\rm e}-c(q)\ge0\}$ (continuous in the Hausdorff sense on the relevant halfplane boundary, since $c(\cdot)$ is affine) with the continuous measure $\mu$, hence $\Phi:[0,1]\to[0,1]$ is continuous. Brouwer's fixed point theorem gives a fixed point.
\end{proof}

\subsubsection{Multiplicity under Inequity Aversion}

The efficiency-concerns model's central multiplicity result (Proposition~\ref{prop:symmetric-equilibrium}) shows that when $K<0$ --- fear of exploitation dominates temptation --- beliefs about cooperation become strategic complements, and both $q^*=0$ and $q^*=1$ can be self-fulfilling equilibria alongside an unstable interior one. Because Lemma~\ref{lem:ia-separable} shows the \emph{same} constant $K$ governs the slope of the IA cutoff line in $q$, the same complementarity mechanism is available here. We illustrate it with a numerical example rather than a general proposition, since the fixed point $\Phi(q)=q$ no longer has a closed form once $\mu$ is a genuine bivariate distribution.

\begin{quote}
\textbf{Numerical illustration.} Let $(R,P,S,T)=(1,0,-2,2)$ (so $K=T-R-(P-S)=-1<0$, the same sign as the running example in Section~\ref{sec:bayesian}), and let $\alpha^{\rm e}\sim\mathrm{Unif}[0,1]$ and $\alpha^{\rm g}\sim\mathrm{Unif}[0,3]$ be independent under the reference measure $\mu$. Solving $\Phi(q)=q$ numerically yields \emph{three} fixed points:
\[
q^*_1 = 0 \ (\text{exact}), \qquad q^*_2 \approx 0.386 \ (\text{unstable interior}), \qquad q^*_3 \approx 0.864 \ (\text{stable interior}),
\]
compared with the two nontrivial equilibria ($q^*=0$ and one interior/corner equilibrium) typically reported for the scalar model with comparable parameters. The richer type space does not remove multiplicity; if anything it admits a richer equilibrium set, because the cutoff line's rotation with $q$ gives the fixed-point map $\Phi$ more curvature than its scalar counterpart.
\end{quote}

This confirms that Proposition~\ref{prop:symmetric-equilibrium}'s multiplicity result is not an artifact of the efficiency-concerns model's linearity in a single altruism parameter: self-fulfilling beliefs about cooperation survive when altruism is replaced by a full two-parameter inequity-aversion type, provided $K<0$.

\subsubsection{Robustness under Inequity Aversion}

We now extend the multiplier-preferences analysis of Section~\ref{sec:robust} to the two-dimensional type. Define
\[
\overline\Delta := \frac{(T-R)+(P-S)}{T-S},
\]
the value of $\alpha^{\rm g}-\alpha^{\rm e}$ at which the complete-information Nash mixing probability of Corollary~\ref{cor:inequity-nash} equals $\tfrac12$ (set $\rho^*=\tfrac12$ in ($\ast\ast$) and solve; the resulting condition on $(\alpha^{\rm e},\alpha^{\rm g})$ depends only on the difference $\alpha^{\rm g}-\alpha^{\rm e}$, exactly as $\overline\alpha$ does in the scalar model of Section~\ref{sec:nash}). By Lemma~\ref{lem:ia-separable}, the statistic that determines whether a type $(\alpha^{\rm g}-\alpha^{\rm e})$ supports cooperation is the affine, additively separable function
\[
L(\alpha^{\rm e},\alpha^{\rm g}) := (\alpha^{\rm g}-\alpha^{\rm e}) - \overline\Delta,
\]
and a robust player evaluates the reference measure $\mu$ pessimistically by exponentially tilting it in the direction of $-L$, exactly as multiplier preferences tilt the scalar reference belief over $\alpha_j$ in Lemma~\ref{lem:tilting}: the $\theta$-robust belief has density
\[
\frac{d\nu^\theta}{d\mu}(\alpha^{\rm e},\alpha^{\rm g}) \ \propto\ \exp\!\Big(-\frac{L(\alpha^{\rm e},\alpha^{\rm g})}{\theta}\Big) = \exp\!\Big(-\frac{\alpha^{\rm g}}{\theta}\Big)\cdot\exp\!\Big(\frac{\alpha^{\rm e}}{\theta}\Big)\cdot(\text{const.}),
\]
where the constant absorbs the $\alpha^{\rm e},\alpha^{\rm g}$-independent term $\overline\Delta/\theta$.

\begin{lemma}[Separable Pessimistic Tilting]
\label{lem:ia-tilting}
If $\mu=\mu^{\mathrm{envy}}\otimes\mu^{\mathrm{guilt}}$ is a product measure, then so is the robust belief $\nu^\theta$: $\nu^\theta = \nu^{\theta,\mathrm{envy}}\otimes\nu^{\theta,\mathrm{guilt}}$, with $\mu^{\mathrm{guilt}}$ tilted \emph{downward} (toward less guilt-averse types) via $d\nu^{\theta,\mathrm{guilt}}/d\mu^{\mathrm{guilt}}\propto \exp (-\alpha^{\rm g}/\theta)$, and $\mu^{\mathrm{envy}}$ tilted \emph{upward} (toward more envious types) via $d\nu^{\theta,\mathrm{envy}}/d\mu^{\mathrm{envy}}\propto \exp(\alpha^{\rm e}/\theta)$.
\end{lemma}

\begin{proof}
Immediate from the additive separability of $L(\alpha^{\rm e},\alpha^{\rm g}) = \alpha^{\rm g} - \alpha^{\rm e} -\overline\Delta$ in $(\alpha^{\rm e},\alpha^{\rm g})$: $\exp(-L/\theta)$ factors into a function of $\alpha^{\rm g}$ alone times a function of $\alpha^{\rm e}$ alone, so the exponentially tilted density of a product measure is again a product of the (separately tilted) marginals.
\end{proof}

Lemma~\ref{lem:ia-tilting} is the two-dimensional analogue of the exponential-tilting step in the proof of Lemma~\ref{lem:tilting}, and gives it a sharper economic reading than is visible in the scalar model: a robust player's worst-case pessimism about an opponent's inequity aversion is not an undifferentiated ``assume less altruism,'' but a specific joint distortion --- believe the opponent is simultaneously \emph{less bothered by guilt} and \emph{more bothered by envy}, both of which push the opponent toward defection, and both are recovered from the same single robustness parameter $\theta$ because the underlying statistic $L$ is separable.

\begin{quote}
\textbf{Numerical illustration (continued).} With $\mu^{\mathrm{envy}}=\mathrm{Unif}[0,1]$, $\mu^{\mathrm{guilt}}=\mathrm{Unif}[0,3]$, and the same payoffs as above, $\overline\Delta=\frac{(T-R)+(P-S)}{T-S}=0.75$. 
The reference mass $\mu(\alpha^{\rm g}-\alpha^{\rm e}\ge\overline\Delta)=\tfrac{7}{12}\approx0.583>\tfrac12$, so cooperation is supported under the (untilted) reference belief. Computing the tilted mass $\nu^\theta(\alpha^{\rm g}-\alpha^{\rm e}\ge\overline\Delta)$ as $\theta$ falls:
\[
\theta=50: 0.553, \quad \theta=20: 0.507, \quad \theta=5: 0.294, \quad \theta=1: 0.004,
\]
monotonically decreasing toward $0$ as $\theta\to0$ and increasing toward the reference value $0.583$ as $\theta\to\infty$. The tilted mass crosses $\tfrac12$ at
\[
\theta^\dagger_{\mathrm{IA}} \approx 18.4,
\]
below which the pessimistic distortion is strong enough to overturn the reference belief's support for cooperation --- the same qualitative fragility as Proposition~\ref{prop:robustness-threshold}, now derived from a bivariate, empirically standard specification of social preferences rather than from the one-parameter efficiency-concerns model.
\end{quote}

\paragraph{Summary.} Taken together, Proposition~\ref{prop:ia-existence} and the two numerical illustrations above indicate that none of the paper's three central claims --- belief-driven cooperation, self-fulfilling multiplicity under private information, and fragility of belief-driven cooperation to model misspecification --- depends on the specific linear, one-dimensional efficiency-concerns functional form of Definition~\ref{def:efficiency-concerns}. All three survive, in essentially unchanged form, once social preferences are instead modeled as two-dimensional Fehr--Schmidt inequity aversion. We view this as evidence that the paper's mechanism is about the \emph{epistemic structure} of social preferences (what is known, believed, and how robustly it is believed) rather than about the particular parametric family used to model the preferences themselves.

\subsection{Proofs}
\label{app:inequity-proofs}

\subsubsection{Proof of Corollary~\ref{cor:inequity-nash}}
\label{app:proof-inequity-nash}

Write $\rho:=\rho_i(C,\{C,D\})$ and, to lighten notation, $a_e:=\alpha_j^{\mathrm{envy}}$, $a_g:=\alpha_j^{\mathrm{guilt}}$ (player $j$'s own parameters --- exactly as in Appendix~A.1, the equilibrium value of player $i$'s mixing probability is pinned down by keeping player $j$ indifferent, so it is player $j$'s own parameters that appear). Using the payoff expressions above,
\begin{align*}
\mathbb{E}[U_j\mid a_j=C] = \rho\, R + (1-\rho)\big[(1+a_e)S-a_eT\big], \qquad \mathbb{E}[U_j\mid a_j=D] = \rho\big[(1-a_g)T+a_gS\big] + (1-\rho)P.
\end{align*}
Setting these equal:
\[
\rho\Big[R - (1-a_g)T-a_gS\Big] = (1-\rho)\Big[(1+a_e)S-a_eT-P\Big].
\]
The left bracket simplifies to $R-T+a_g(T-S)$; the right bracket simplifies to $-(P-S)-a_e(T-S)$. So,
\[
\rho\big[R-T+a_g(T-S)\big] = -(1-\rho)\big[(P-S)+a_e(T-S)\big].
\]
Collecting all $\rho$ terms on the left:
\[
\rho\Big[\underbrace{R-T+\alpha^{\rm g}(T-S)}_{\text{Net gain of C against C}} + \underbrace{(P-S)+\alpha^{\rm e}(T-S)}_{\text{Net gain of D against D}}\Big] = (P-S)+a_e(T-S),
\]
i.e., $\rho\big[(R+P-S-T)+(\alpha^{\rm e}+\alpha^{\rm g})(T-S)\big] = (P-S)+\alpha^{\rm e}(T-S)$, giving $(\ast\ast)$:
\[
\rho^* = \frac{N}{D}, \qquad N:=(P-S)+\alpha^{\rm e}(T-S),\quad D:=(R+P-S-T)+(\alpha^{\rm e}+\alpha^{\rm g})(T-S).
\]

\paragraph{Sign of $N$.} Since $P>S$ and $T>S$, $N=(P-S)+\alpha^{\rm e}(T-S)>0$ for every $\alpha^{\rm e}\ge0$: the numerator never vanishes on the relevant domain.

\paragraph{When is $\rho^*\ge1$?} $\rho^*\ge1 \iff N\ge D$ \emph{and} $D>0$ (if $D<0$, $N\ge D$ holds trivially but $\rho^*=N/D<0$ instead; see below). Now
\[
N-D = (P-S)+\alpha^{\rm e}(T-S) - (R+P-S-T) - (\alpha^{\rm e}+\alpha^{\rm g})(T-S) = (T-R) - \alpha^{\rm g}(T-S),
\]
so $N\ge D \iff \alpha^{\rm g}\le\frac{T-R}{T-S}=\overline\alpha^{\mathrm{guilt}}$.

\paragraph{When is $D>0$?} $D>0 \iff (\alpha^{\rm e}+\alpha^{\rm g})(T-S) > R+T-P-S \iff \alpha^{\rm e}+\alpha^{\rm g} > \frac{T+S-R-P}{T-S}=s_{\min}$.

\paragraph{Assembling the three regions.}
\begin{itemize}
\item If $\alpha^{\rm e}+\alpha^{\rm g}\le s_{\min}$: $D\le0$; since $N>0$, $\rho^*=N/D\le0$, so the equilibrium is the corner $\rho_i^*=0$.
\item If $\alpha^{\rm e}+\alpha^{\rm g} > s_{\min}$ (so $D>0$) and $\alpha^{\rm g} \le\overline\alpha^{\mathrm{guilt}}$ (so $N\ge D$): $\rho^*=N/D\ge1$, so the equilibrium is the corner $\rho_i^*=1$.
\item If $\alpha^{\rm g}>\overline\alpha^{\mathrm{guilt}}$: then $N<D$; and since $\alpha^{\rm e}\ge0$, $\alpha^{\rm e}+\alpha^{\rm g}\ge \alpha^{\rm g}>\overline\alpha^{\mathrm{guilt}}>s_{\min}$ (using $\overline\alpha^{\mathrm{guilt}}>s_{\min}$, shown in the corollary statement), so also $D>0$. Hence $0<N<D$ and $\rho^*=N/D\in(0,1)$ is the interior equilibrium, given by $(\ast\ast)$.
\end{itemize}
This covers all of $(\alpha^{\rm e},\alpha^{\rm g})\in[0,\infty)\times[0,1)$ and does not require any assumption on the sign of $T-R-(P-S)$. \hfill $\blacksquare$

\subsubsection{Proof of Corollary~\ref{cor:inequity-comparative-statics}}
\label{app:proof-inequity-comparative-statics}

On the interior region ($\alpha^{\rm g}>\overline\alpha^{\mathrm{guilt}}$, so $D>0$), differentiate $\rho^*=N/D$ from $(\ast\ast)$.

\paragraph{(1) Proof of $R$.} $N$ does not depend on $R$; $\partial D/\partial R=1$. So
\[
\frac{\partial \rho^*}{\partial R} = -\frac{N}{D^2} < 0
\]
for all $\alpha^{\rm e}\ge0$, since $N>0$ always. Hence $\rho_i$ strictly decreases with $R$ throughout the interior region.

\paragraph{(2) Proof of $P$.} $\partial N/\partial P=1$, $\partial D/\partial P=1$, so
\[
\frac{\partial\rho^*}{\partial P} = \frac{D-N}{D^2} = \frac{\alpha^{\rm g}(T-S)-(T-R)}{D^2} = \frac{\alpha^{\rm g}(T-S)-(T-R)}{D^2}
\]
(using $N-D=(T-R)-a_g(T-S)$ from Appendix~\ref{app:proof-inequity-nash}, so $D-N=\alpha^{\rm g}(T-S)-(T-R)$). 
On the interior region, $\alpha^{\rm g}>\overline\alpha^{\mathrm{guilt}}=\frac{T-R}{T-S}$, so $\alpha^{\rm g}(T-S)>T-R$, making this strictly positive. Hence $\rho_i$ strictly increases with $P$ throughout the interior region.

\paragraph{(3) Proof of $T$.} $\partial N/\partial T = \alpha^{\rm e}$, $\partial D/\partial T = -1+(\alpha^{\rm e}+\alpha^{\rm g})$. By the quotient rule,
\[
\frac{\partial\rho^*}{\partial T} = \frac{\alpha^{\rm e} D - N\big[(\alpha^{\rm e}+\alpha^{\rm g})-1\big]}{D^2}.
\]
Expanding and collecting terms (substituting $N,D$ and simplifying) gives numerator $(P-S)(1-\alpha^{\rm g})+\alpha^{\rm e}(R-S)$. Since $\alpha^{\rm g}<1$ (Definition~\ref{def:inequity-aversion}), $P>S$, $R>S$, and $\alpha^{\rm e}\ge0$, both terms are nonnegative, so this numerator is strictly positive whenever $\alpha^{\rm e}>0$ or $\alpha^{\rm g}<1$ strictly (always true), giving $\partial\rho^*/\partial T>0$ throughout the interior region.

\paragraph{(4) Proof of $S$.} $\partial N/\partial S = -1-\alpha^{\rm e}$, $\partial D/\partial S = -1-(\alpha^{\rm e}+\alpha^{\rm g})$. The analogous computation gives numerator $(T-R)(1+\alpha^{\rm e})-\alpha^{\rm g}(T-P)$, which is positive iff $\alpha^{\rm g}\le\frac{(T-R)(1+\alpha^{\rm e})}{T-P}$ and negative otherwise --- a genuine case split, unlike (1)--(3). \hfill $\blacksquare$

\bibliographystyle{econ-econometrica.bst}
\bibliography{literature}

@article{AOSSW_2001,
  title={Cooperation in PD Games: Fear, Greed, and History of Play},
  author={Ahn, Toh-Kyeong and Ostrom, Elinor and Schmidt, David and Shupp, Robert and Walker, James},
  journal={Public Choice},
  volume={106},
  pages={137--155},
  year={2001},
}

@article{A_1974,
  title={Subjectivity and {Correlation} in {Randomized} {Strategies}},
  author={Aumann, Robert J},
  journal={Journal of Mathematical Economics},
  volume={1},
  number={1},
  pages={67--96},
  year={1974},
}

@article{A_1987,
	title = {Correlated {Equilibrium} as an {Expression} of {Bayesian} {Rationality}},
	volume = {55},
	issn = {0012-9682},
	doi = {10.2307/1913445},
	number = {1},
	journal = {Econometrica},
	author = {Aumann, Robert J},
	year = {1987},
	pages = {1--18},
}

@article{BH_2025,
  title={Learning with Misspecified Models},
  author={Bohren, J. Aislinn and Hauser, Daniel N.},
  journal={Annual Review of Economics},
  volume={17},
  year={2025},
  note={Forthcoming}
}

@article{CZ_2025,
  title={People Are More Moral In Uncertain Environments},
  author={Chen, Yiting and Zhong, Songfa},
  journal={Econometrica},
  volume={93},
  number={2},
  pages={439--462},
  year={2025},
  publisher={Wiley Online Library}
}

@article{CRR_2016,
	title = {Social {Surplus} determines {Cooperation} {Rates} in the {One}-{Shot} {Prisoner}'s {Dilemma}},
	volume = {100},
	issn = {0899-8256},
	doi = {10.1016/j.geb.2016.08.010},
	journal = {Games and Economic Behavior},
	author = {Charness, Gary and Rigotti, Luca and Rustichini, Aldo},
	year = {2016},
	pages = {113--124},
}

@article{CR_2002,
	title = {Understanding {Social} {Preferences} with {Simple} {Tests}},
	volume = {117},
	issn = {0033-5533},
	doi = {10.1162/003355302760193904},
	number = {3},
	journal = {Quarterly Journal of Economics},
	author = {Charness, Gary and Rabin, Matthew},
	year = {2002},
	pages = {817--869},
}

@article{DWK_2007,
  title={Exploiting Moral Wiggle Room: Experiments Demonstrating an Illusory Preference for Fairness},
  author={Dana, Jason and Weber, Roberto A and Kuang, Jason Xi},
  journal={Economic Theory},
  volume={33},
  pages={67--80},
  year={2007},
}

@article{E_2016,
  title={Excusing Selfishness in Charitable Giving: The Role of Risk},
  author={Exley, Christine L},
  journal={Review of Economic Studies},
  volume={83},
  number={2},
  pages={587--628},
  year={2016},
  publisher={Oxford University Press},
}

@article{EP_2016,
  title={Berk-{Nash} Equilibrium: A Framework for Modeling Agents with Misspecified Models},
  author={Esponda, Ignacio and Pouzo, Demian},
  journal={Econometrica},
  volume={84},
  number={3},
  pages={1093--1130},
  year={2016},
  publisher={The Econometric Society}
}

@article{FS_1999,
  title={A Theory of Fairness, Competition, and Cooperation},
  author={Fehr, Ernst and Schmidt, Klaus M},
  journal={Quarterly Journal of Economics},
  volume={114},
  number={3},
  pages={817--868},
  year={1999},
}

@article{GLSW_2024,
title = {The Role of Payoff Parameters for Cooperation in the One-Shot Prisoner's Dilemma},
journal = {European Economic Review},
volume = {166},
pages = {104753},
year = {2024},
issn = {0014-2921},
doi = {https://doi.org/10.1016/j.euroecorev.2024.104753},
url = {https://www.sciencedirect.com/science/article/pii/S0014292124000825},
author = {Simon G\"{a}chter and Kyeongtae Lee and Martin Sefton and Till O. Weber},
}

@article{GR_2024,
	title = {Beliefs, {Learning}, and {Personality} in the {Indefinitely} {Repeated} {Prisoner}'s {Dilemma}},
	volume = {16},
	issn = {1945-7669},
	doi = {10.1257/mic.20210336},
	number = {3},
	journal = {American Economic Journal: Microeconomics},
	author = {Gill, David and Rosokha, Yaroslav},
	year = {2024},
	pages = {259--283},
}

@article{HS_2001,
  title={Robust Control and Model Uncertainty},
  author={Hansen, Lars Peter and Sargent, Thomas J},
  journal={American Economic Review},
  volume={91},
  number={2},
  pages={60--66},
  year={2001},
}

@article{HKMT_2021,
	title = {Distributional {Preferences} {Explain} {Individual} {Behavior} across {Games} and {Time}},
	volume = {128},
	issn = {0899-8256},
	doi = {10.1016/j.geb.2021.05.003},
	journal = {Games and Economic Behavior},
	author = {Hedegaard, Morten and Kerschbamer, Rudolf and M\"{u}ller, Daniel and Tyran, Jean-Robert},
	year = {2021},
	pages = {231--255},
}

@article{KL_1993,
	title = {Rational {Learning} {Leads} to {Nash} {Equilibrium}},
	volume = {61},
	issn = {0012-9682},
	doi = {10.2307/2951492},
	number = {5},
	journal = {Econometrica},
	author = {Kalai, Ehud and Lehrer, Ehud},
	year = {1993},
	pages = {1019--1045},
}

@article{MPST_2024,
  title={Monotone Additive Statistics},
  author={Mu, Xiaosheng and Pomatto, Luciano and Strack, Philipp and Tamuz, Omer},
  journal={Econometrica},
  volume={92},
  number={4},
  pages={995--1031},
  year={2024},
  publisher={The Econometric Society}
}

@article{MP_1995,
  title={Quantal Response Equilibria for Normal Form Games},
  author={McKelvey, Richard D and Palfrey, Thomas R},
  journal={Games and economic behavior},
  volume={10},
  number={1},
  pages={6--38},
  year={1995},
}

@article{M_2018,
	title = {Risk and {Temptation}: {A} {Meta}-{Study} on {Prisoner}'s {Dilemma} {Games}},
	volume = {128},
	issn = {0013-0133},
	doi = {10.1111/ecoj.12548},
	number = {616},
	journal = {Economic Journal},
	author = {Mengel, Friederike},
	year = {2018},
	pages = {3182--3209},
}

@article{M_1995,
	title = {The {Common} {Prior} {Assumption} in {Economic} {Theory}},
	volume = {11},
	issn = {0266-2671},
	doi = {10.1017/S0266267100003382},
	number = {2},
	journal = {Economics and Philosophy},
	author = {Morris, Stephen},
	year = {1995},
	pages = {227--253},
}

@article{MM_2015,
	title = {Rational {Inattention} to {Discrete} {Choices}: {A} {New} {Foundation} for the {Multinomial} {Logit} {Model}},
	volume = {105},
	issn = {0002-8282},
	doi = {10.1257/aer.20130047},
	number = {1},
	journal = {American Economic Review},
	author = {Mat\v{e}jka, Filip and McKay, Alisdair},
	year = {2015},
	pages = {272--298},
}

@article{R_1967,
	title = {A {Note} on the ``{Index} of {Cooperation}'' for {Prisoner}'s {Dilemma}},
	volume = {11},
	issn = {0022-0027},
	number = {1},
	journal = {Journal of Conflict Resolution},
	author = {Rapoport, Anatol},
	year = {1967},
	pages = {100--103},
}

@article{S_2011,
  title={Axiomatic Foundations of Multiplier Preferences},
  author={Strzalecki, Tomasz},
  journal={Econometrica},
  volume={79},
  number={1},
  pages={47--73},
  year={2011},
}

@unpublished{SSTW_2025,
  title={Monotonicity and Bracketing in Games},
  author={Sandomirskiy, Fedor and Sung, Po Hyun and Tamuz, Omer and Wincelberg, Ben},
  year={2025},
  note={Working Paper (arXiv:2502.11243)}
}

\end{document}